\documentclass[11pt]{article}
\usepackage{fullpage}
\usepackage{graphicx}
\usepackage{upref}
\usepackage{enumerate}
\usepackage{latexsym}
\usepackage{pdfpages}
\usepackage{tikz}
\usetikzlibrary{shapes, arrows}
\usepackage[bookmarks]{hyperref}
\usepackage{sansmath}

\usepackage{algorithm}
\usepackage[noend]{algpseudocode}

\usepackage{pdflscape}

\usepackage{color,graphics}
\usepackage{comment} 
\usepackage{caption}
\usepackage{subcaption}
\usepackage{wrapfig}
\usepackage{braket}
\usepackage{amssymb}
\usepackage{amsmath}
\usepackage{amsthm}
\usepackage{mathrsfs}
\usepackage{mathtools}
\usepackage{commath}
\usepackage{thmtools,thm-restate}
\usepackage{tabularx, makecell, multirow}
\usepackage{tikz}
\usetikzlibrary{decorations.pathreplacing,decorations.pathmorphing}
\usepackage{boxedminipage}
\usetikzlibrary{patterns}
\usetikzlibrary{arrows.meta, positioning}

\DeclareMathOperator{\TypeO}{\mathsf{Type-O}}

\usepackage{amsfonts}
\usepackage{varioref}
\usepackage[ansinew]{inputenc}
\usepackage[many]{tcolorbox}

\usepackage{xcolor}
\hypersetup{
	colorlinks,
	linkcolor={red!75!black},
	citecolor={blue!75!black},
	urlcolor={blue!75!black}
}
\usepackage[hyperpageref]{backref}

\newtheorem{theorem}{Theorem}[section]
\newtheorem{lemma}[theorem]{Lemma}
\newtheorem{corollary}[theorem]{Corollary}
\newtheorem{conjecture}[theorem]{Conjecture}
\newtheorem{claim}[theorem]{Claim}

\newtheorem{defn}[theorem]{Definition}
\newtheorem{remark}[theorem]{Remark}

\theoremstyle{definition}

\AtBeginDocument{}

\newcommand{\E}{\mathbb{E}} 	

\newcommand{\poly}{\mathrm{poly}}
\newcommand{\polylog}{\mathrm{polylog}}

\newcommand{\bpSAT}{\textsc{BPSAT}}
\newcommand{\circuitSAT}{\texttt{Circuit}\textsc{SAT}}

\newcommand{\CNFSAT}{\textsc{CNFSAT}}
\newcommand{\acQSETH}{$\textsc{AC}^0_2$\textsc{-QSETH}}

\newcommand{\QSETH}{\textsc{QSETH}}
\newcommand{\SETH}{\textsc{SETH}}
\newcommand{\countingSETH}{\# \textsc{SETH}}
\newcommand{\paritySETH}{\oplus \textsc{SETH}}
\newcommand{\circuitSETH}{\texttt{Circuit}\textsc{SETH}}
\newcommand{\NCSETH}{\textsc{NCSETH}}

\renewcommand{\P}{\mathsf{P}}
\newcommand{\PSPACE}{\mathsf{PSPACE}}

\newcommand{\CNF}{\textsc{CNF}}
\newcommand{\DNF}{\textsc{DNF}}
\newcommand{\AC}{\mathsf{AC}_{2}^{0}}
\newcommand{\ACconstDepth}{\mathsf{AC^0}}

\newcommand{\ACarg}[2]{\mathsf{AC}_{{#1},{#2}}^{0}}
\newcommand{\NC}{\mathsf{NC}}
\newcommand{\NCarg}[2]{\mathsf{NC}_{#1,#2}}
\newcommand{\Size}[1]{\mathsf{SIZE}(#1)}
\newcommand{\SizeArg}[2]{\mathsf{SIZE}_{#1,#2}}
\newcommand{\circuit}[1]{\mathsf{#1}}
\newcommand{\Ppoly}{\mathsf{P}/\textnormal{poly}}

\newcommand{\BQTIME}[1]{\textsc{BQTIME}(#1)}
\newcommand{\BQP}{\textsc{BQP}}
\newcommand{\PRF}{\mathrm{PRF}}

\newcommand{\propertyP}{\textsc{P}}
\newcommand{\propertyD}{\textsc{D}}
\newcommand{\propertyOR}{\textsc{OR}}
\newcommand{\propertyAND}{\textsc{AND}}

\newcommand{\propertyParity}{\textsc{parity}}
\newcommand{\propertyMajority}{\textsc{majority}}
\newcommand{\propertyOddMajority}{\textit{odd}\textsc{-majority}}
\newcommand{\propertyStrictMajority}{\textit{st}\textsc{-majority}}

\renewcommand{\epsilon}{\varepsilon}

\newcommand{\probabilityOf}{\mathrm{Pr}}

\newcommand{\rnote}[1]{}

\newcommand{\zo}{\{0,1\}}

\newcommand{\la}{\leftarrow}

\newcommand{\Q}{\mathsf{Q}}

\usepackage[capitalize,nameinlink,noabbrev]{cleveref}
\crefname{conjecture}{Conjecture}{Conjectures}
\crefname{defn}{Definition}{definition}
\crefname{claim}{Claim}{Claims}
\creflabelformat{equation}{#2\textup{#1}#3}

\newcommand{\truthtable}[1]{\texttt{tt} #1}
\newcommand{\setOfAllStrings}{F}

\newcommand{\concatenate}{\bigcirc}
\newcommand{\DTIME}[1]{\textsc{DTIME}(#1)}
\newcommand{\linear}{\texttt{linear}}
\newcommand{\fourierCoefficients}[2]{\ensuremath{\Hat{\textsc{#1}}(#2)}}
\newcommand{\TR}{\mathsf{TR}}

\title{Fine-Grained Complexity via Quantum Natural Proofs}

\author{Yanlin Chen\thanks{University of Maryland (QuICS) {\tt yanlin@umd.edu}}
\and 
Yilei Chen\thanks{Tsinghua University, Shanghai Qi Zhi Institute {\tt chenyilei.ra@gmail.com}. Supported by Tsinghua University startup funding, and Shanghai Qi Zhi Institute Innovation Program SQZ202405. }
\and 
Rajendra Kumar\thanks{Indian Institute of Technology Delhi {\tt rajendra@cse.iitd.ac.in}. Supported by Chandruka New Faculty Fellowship at IIT Delhi.}
\and 
Subhasree Patro\thanks{Eindhoven University of Technology {\tt patrofied@gmail.com}}
\and 
Florian Speelman\thanks{University of Amsterdam and QuSoft {\tt f.speelman@uva.nl}. Supported by the Dutch Ministry of Economic Affairs and Climate Policy (EZK), as part of the Quantum Delta NL program, and the project Divide and Quantum `D\&Q' NWA.1389.20.241 of the program `NWA-ORC', which is partly funded by the Dutch Research Council (NWO).}}

\date{}

\begin{document}
\pagenumbering{roman}

\maketitle

\begin{abstract}
Buhrman, Patro, and Speelman \cite{BPS21} presented a framework of conjectures that together form a quantum analogue of the strong exponential-time hypothesis and its variants. They called it the QSETH framework. In this paper, using a notion of quantum natural proofs (built from natural proofs introduced by Razborov and Rudich), we show how part of the QSETH conjecture that requires properties to be `compression oblivious' can in many cases be replaced by assuming the existence of quantum-secure pseudorandom functions, a standard hardness assumption. Combined with techniques from Fourier analysis of Boolean functions, we show that properties such as $\propertyParity$ and $\propertyMajority$ are compression oblivious for certain circuit class $\Lambda$ if subexponentially secure quantum pseudorandom functions exist in $\Lambda$, answering an open question in \cite{BPS21}.
\end{abstract}

\thispagestyle{empty}
\newpage
\setcounter{tocdepth}{2}
{
  \hypersetup{linkcolor=black}
  \tableofcontents
}
\newpage
\pagenumbering{arabic}

\section{Introduction}
\label{sec:Introduction}

A well-studied problem in complexity theory is $\CNFSAT$, the problem of whether a formula, input in conjunctive normal form, has a satisfying assignment. Not only is there no polynomial time known for $\CNFSAT$, there is also no known algorithm that significantly improves over the brute-force method of checking all $2^n$ assignments. Impagliazzo, Paturi and Zane in \cite{IP01,IPZ01} conjectured that determining whether an input $\CNF$ formula is satisfiable or not cannot be done in time $O(2^{n(1-\epsilon)})$ for any constant $\epsilon>0$. They called it \emph{Strong Exponential-Time Hypothesis} (SETH). As a consequence, $\SETH$-based lower bounds were shown for many computational problems; for example, see \cite{williams2005new,BI15,HNS20}.

\paragraph{Variants of classical $\SETH$} Subsequent works also explored alternate directions in which the strong exponential-time hypothesis can be weakened and thereby made more plausible. For, e.g.,
\begin{enumerate}
    \item \label{item:VariantsOfSatisfiability} The \emph{satisfiability} problem is studied for other (more general) succinct representations of Boolean functions such as $\bpSAT$ or $\circuitSAT$ - the problem of whether a polynomial-size branching program or a polynomial-size circuit has a satisfying assignment. Furthermore, respective hardness conjectures for these satisfiability problems, namely $\NCSETH$ and $\circuitSETH$ (under the name `$\circuitSAT$ form of $\SETH$') have also been discussed together with their implications \cite{AHWW16,Williams24CircuitSAT}. 
    \item \label{item:VariantsOfProperties} Another direction of study that has made $\SETH$ more plausible is by computing more complicated properties of the formula. For example, counting the number of satisfying assignments and computing whether the number of satisfying assignments is even or odd is captured by $\countingSETH$ and $\paritySETH$, respectively; see \cite{CDLMNOPSW16} for example.
\end{enumerate}
Unfortunately, neither $\SETH$ nor $\NCSETH$ or $\circuitSETH$, hold relative to \emph{quantum} computation, as using Grover's algorithm for unstructured search~\cite{grover1996fast} it is possible to solve the respective satisfiability problems in $O(2^{n/2})$ time - leading to conjectured lower bounds of complexity $\Omega(2^{n/2})$. While these lower bounds have interesting consequences (as shown by \cite{BasicQSETH-Aaronson-2020,BPS21}), in the quantum case, generalizing the properties to be computed (as suggested in \Cref{item:VariantsOfProperties}) is a way to come up with variants of the hypothesis that enable stronger bounds: for many of such tasks it is likely that the
quadratic quantum speedup, as given by Grover's algorithm, no longer exists. In fact, motivated by this exact observation, Buhrman, Patro, and Speelman \cite{BPS21} introduced a framework of Quantum Strong Exponential-Time Hypotheses ($\QSETH$) as quantum analogues to $\SETH$, with a striking feature that allows one to technically unify quantum analogues of $\SETH$, $\oplus\SETH$, $\#\SETH$, $\NCSETH$, $\oplus \NCSETH$, etc., all under one umbrella conjecture.

\paragraph{The $\QSETH$ framework} Buhrman \textit{et al.}\ consider the problem in which one is given a formula or a circuit representation of a Boolean function $f:\{0,1\}^n \rightarrow \{0,1\}$ and is asked whether a property $\propertyP: \{0,1\}^{2^n}\rightarrow \{0,1\}$ on its truth table\footnote{Truth table of a formula/circuit $\phi$ on $n$ variables, denoted by $tt(\phi)$, is a $2^n$ bit string derived in the following way $tt(\phi)=\bigcirc_{a \in \{0,1\}^n}\phi(a)$; the symbol $\circ$ denotes concatenation.} evaluates to $1$. They conjectured that when the circuit representation is \emph{obfuscated} enough, then for \textit{most} properties $\propertyP$ (that are `compression oblivious' with respect to the representation), the time taken to compute $\propertyP$ on the truth table of $\poly(n)$-sized circuits is lower bounded by $\Q(\propertyP)$, i.e., the $1/3$-bounded error quantum query complexity of $\propertyP$, on all bit strings of length $2^n$.

Note that there are two requirements to using the \QSETH{} framework. Firstly, the input formula or circuit must be `obfuscated enough' - this translates to saying that knowing the description of the input must not help towards computing $\propertyP$ any more than just having query access to the truth table of the input. The second requirement (which is the main topic of discussion of this paper) is that the property $\propertyP$ itself must be `compression oblivious' - a notion introduced in \cite{BPS21} that we discuss next.

\paragraph{Compression-oblivious properties} 
Buhrman \textit{et al.}\ defined this notion to capture properties whose \textit{query complexity} is a lower bound for the \emph{time complexity} to compute the property even for inputs that are ``compressible'' as a truth table of small formulas (or circuits). This notion is useful (to their framework) because if the formula (or circuit) is obfuscated enough then one can use the quantum query lower bound of such a property to conjecture a time lower bound of computing that property on a small formula (or circuit).\footnote{Here is an example to illustrate why obfuscation is important. Consider the $\propertyOR$ property. The quantum query complexity of $\propertyOR$ when restricted to truth tables of small $\CNF$ or $\DNF$ formulas defined on $n$ input variables is $\Omega(\sqrt{N})$ where $N=2^n$ - let us denote this statement by $(\star)$. Furthermore, the structure of $\CNF$ formulas are not known to help in computing $\propertyOR$ more efficiently than $O(\sqrt{N})$. Therefore, we can use the same lower bound from $(\star)$ to conjecture the lower bound of $\Omega(\sqrt{N})$ for $\CNFSAT$ in the quantum setting. However, $\DNF$ formulas are not obfuscated enough for $\propertyOR$ as we can compute $\propertyOR$ in $\poly(n)$ time when input formulas are $\DNF$s. This means we cannot use the query lower bound of $\sqrt{N}$ from $(\star)$ for the $\DNF$s as we do in the case of $\CNF$ formulas.} Therefore, before using their framework, one has to be convinced that the property is indeed compression-oblivious with respect to the circuit representation considered.
They give various results to initiate the study of the set of compression-oblivious languages. For example, they show the following.
\begin{itemize}
    \item For simple properties, such as $\propertyOR, \propertyAND$, they were able to prove that these are compression-oblivious for polynomial-size $\AC$ circuits.
    \item For more complicated properties, such as $\propertyParity$ or $\propertyMajority$, they show that proving compression-obliviousness would separate $\P$ from $\PSPACE$; see Theorem~9 in \cite{BPS21}. Their result indicates that \emph{unconditionally} proving compression-obliviousness for a complicated property is hard. One way around this would be to directly \emph{conjecture} the compression-obliviousness of such properties, which was the choice of \cite{BPS21}. However, it would be desirable to put this on more solid ground.
\end{itemize}
One might ask --- \textit{why is the notion of compression-oblivious properties useful for the QSETH framework, especially when it is hard to prove results about this notion for many properties? 
} 
\newline

Towards answering this question, consider the $\oplus\QSETH$ hypothesis - implicitly stated and discussed in \cite{BPS21} and explicitly stated in \cite{CCKSS23QSETH,Che24Thesis}. It says,
\begin{conjecture}[Conjecture~8.9 in \cite{Che24Thesis}]
For each constant $\delta >0$, there exists $c>0$ such that there
is no bounded-error quantum algorithm that solves $\oplus \CNFSAT$ (even restricted to formulas with at most $c \cdot n^2$ clauses) in $O(2^{n(1-\delta)})$ time.
\end{conjecture}

One can go about refuting this conjecture in two (immediate) ways.
\begin{enumerate}
    \item \label{item:COmotivation1} Either one could use the description of the input formulas to help deduce $\propertyParity$ of the truth table corresponding to this formula (which would mean that the $\CNF$ formula is not obfuscated enough and one could use the structure of the formula towards computing $\propertyParity$ more efficiently), or
    \item \label{item:COmotivation2} One could use the fact that the set of truth tables on which $\propertyParity$ is to be computed is `not too big' (as the number of polynomial length formulas are not too many) and potentially use this information to speed up the computation quantumly.  
\end{enumerate}
While both approaches seem like good starting points, the likelihood of \Cref{item:COmotivation1} not working out is probably for the same reason why $\oplus \SETH$ is not refuted yet. Which brings us to \Cref{item:COmotivation2}. There is a stark contrast between the possible speedups that quantum can offer over classical computation. For example, there exists a quantum \emph{query} algorithm that can compute \emph{any} property in $\widetilde{O}(\sqrt{N})$ many queries, when the input is promised to be the $N$-length truth table of a small formula \cite{AIKMPY04,AIKRY07,Kothari14}; such an algorithm is impossible in the classical setting. Although efficient in the number of queries, the algorithm itself is not \emph{time} efficient. But the existence of this algorithm (in contrast to the lack of one in the classical setting) does alert us that query complexity cannot naively be used to lower bound time complexity in our setting, without taking care. (The results of the current work will indicate that such a fast algorithm is unlikely to refute $\oplus\QSETH$ via \Cref{item:COmotivation2}.)

As pointed out by Buhrman \textit{et al.}\ \cite{BPS21}, the proof barrier does not allow one to (easily) prove compression-obliviousness of properties such as $\propertyParity$ or $\propertyMajority$. It is therefore natural to ask if there are other reasons why certain properties, say $\propertyParity$ for example, should be compression oblivious. The authors leave the following open question:
\newline

\noindent \textit{`What complexity-theoretic assumptions are needed to show that, e.g., the $\propertyParity$ property, is compression oblivious?'}
\newline

It is precisely this question we answer in this paper. While our answer doesn't give us a way around \emph{conjecturing} the compression obliviousness of these properties, it gives strong evidence supporting the conjecture and the use of the $\QSETH$ framework.

\subsection{Main idea: arguing compression obliviousness via quantum natural proofs}
\label{sec:OurResults}

Following is the answer to the question asked earlier.

\begin{theorem}[Informal restatement of \Cref{thm:LinearPviolatesQPRFassumption}]
\label{thm:MainThmAsStatedInIntro}
The properties $\propertyParity$ and $\propertyMajority$ are compression oblivious for a circuit class $\mathsf{M}$ unless there are no quantum-secure $\mathsf{M}$-constructible $\PRF$s (of a particular type).\footnote{For us, the logical statements `$A$ implies $B$', `if $A$ then $B$' and `$\neg A$ unless $B$' are all equivalent.}
\end{theorem}

In fact, we conclude the same for a larger class of properties (which includes $\propertyParity$ and $\propertyMajority$). The details are mentioned in the statement of \cref{thm:LinearPviolatesQPRFassumption}.


\subsubsection{Natural proofs by Razborov and Rudich}

To prove our results, we use the notion of \emph{natural proofs} introduced by Razborov and Rudich \cite{RR94}. Consider a Boolean function defined on other Boolean functions, or equivalently defined on their truth tables. Razborov and Rudich called a property (of Boolean functions) \emph{natural} if it satisfies certain `constructivity' and `largeness' conditions. Roughly speaking, for their proof, the constructivity condition requires that the property is decidable in $\poly(n)$ time when the $2^n$-sized truth table of an $n$-input Boolean function is given as input. The largeness condition requires that the property holds for a sufficiently large fraction of the set of all Boolean functions. Moreover, they say a property is useful against a complexity class $\Lambda$ if every sequence of Boolean functions having the property (infinitely often) defines a language outside of $\Lambda$. A \emph{natural proof} is a proof that establishes that a certain language lies outside of $\Lambda$. This can be formally defined as (following the notation of Chow~\cite{Cho11}),

\begin{defn}[$\Upsilon$-natural property with density $\delta$, Definition~3 in \cite{Cho11}]
\label{defn:UpsilonNaturalProperty}
Let $\Upsilon$ be a complexity class and let $\delta \coloneqq (\delta_n)_{n \in \mathbb{N}}$ denote a sequence of positive real numbers. A property $D\coloneqq (D_n)_{n \in \mathbb{N}}$ is $\Upsilon$-natural with density $\delta \coloneqq (\delta_n)_{n \in \mathbb{N}}$ if the following statements hold:
\begin{enumerate}
    \item largeness: $\exists n', \forall n>n'$, $|D_n| \geq 2^{2^n} \cdot \delta_n$, and
    \item constructivity: the problem of determining whether a Boolean function $f_n \in D_n$, when given as input the truth table of $f_n$ on $n$ variables, is decidable in $\Upsilon$. 
\end{enumerate}  
\end{defn}
Note that the largeness is captured by the density function $(\delta_n)_{n \in \mathbb{N}}$ and constructivity is captured by the complexity class $\Upsilon$.
Additionally, this $\Upsilon$-natural property is useful against a complexity class $\Lambda$ if the following statement is satisfied.

\begin{defn}[Useful property, Definition~4 in \cite{Cho11}]
\label{defn:Quasi-UsefulProperty}
Let $\Lambda$ denote a complexity class. A property $D \coloneqq (D_n)$ is useful against complexity class $\Lambda$ if $\forall (f_n)$, we have $(f_n) \notin \Lambda$ whenever $(f_n)$ is a sequence of Boolean functions satisfying $f_n \in D_n$ for infinitely many $n$.
\end{defn}

Using these definitions, Razborov and Rudich prove the following.

\begin{theorem}[Main theorem of Razborov and Rudich as stated in \cite{Cho11}]
\label{thm:RazRudAsToldByChow}
If there exists a $2^{k^\epsilon}$-hard pseudorandom generator in $\Size{k^c}$, then there is no $\Ppoly$-constructible natural property with density~$1/2^{n^d}$ that is useful against the complexity class $\Size{n^e}$ whenever $e > 1 + cd/\epsilon$.
\end{theorem}

Broadly speaking, their proof idea is to use the \emph{algorithm} that computes the $\Ppoly$-natural property of density $(\delta_n=\frac{1}{2^{n^d}})$ that is useful against $\Size{n^e}$ (if such a property exists) as a \emph{distinguisher} to distinguish truth tables of functions with $\poly(n)$-size circuits from truth tables of truly random functions. As this algorithm computes a natural property of density $(\delta_n)$, this means the distinguisher (which is the algorithm itself) can now distinguish a $\delta_n$ fraction of truly random functions. As $\delta_n$ here is quite large (for all $n$), this translates to a statistical test with a significant advantage to distinguish pseudorandom functions from truly random ($n$-input) functions. This in turn breaks the assumption about the existence of pseudorandom generators.

\subsubsection{Our results}

We use techniques similar to those of Razborov and Rudich to relate compression obliviousness of $\propertyParity$ (and $\propertyMajority$) with the existence of quantum-secure pseudorandom functions. We divide our contributions into two broad categories: a part that uses natural proofs, and a part that uses techniques from Fourier analysis of Boolean functions. 

\paragraph{Part 1: quantum natural proofs} We use the phrase \emph{quantum natural proofs} to indicate that the proof uses natural properties computable in a complexity class capturing quantum computations. In our case, it is sufficient to consider the complexity class $\BQTIME{\linear}$. Let $\propertyP=(\propertyP_n)_{n \in \mathbb{N}}$ be the property whose compression obliviousness, with respect to $\mathsf{M}$, we want to study. Additionally, suppose that $\propertyP$ satisfies the following two conditions.
\begin{itemize}
    \item \label{item:PropertyInDtime} The property $\propertyP$ is computable in $\DTIME{N}$; where $N$ is the length of the input.
    \item \label{item:PropertyHasHardness} For all $S \subseteq \{0,1\}^N$ of size $|S| \geq 2^N\left(1-\frac{1}{N^{2}}\right)$, the $\Q(\propertyP_n|_{S})$, i.e., the quantum query complexity to compute $\propertyP$ on strings in set $S$, is $\Omega(N)$.
\end{itemize}
We can show that $\propertyP$ has to be compression oblivious with respect to $\mathsf{M}$ unless a type of quantum secure pseudorandom functions doesn't exist.

Towards contradiction, let us assume that $\propertyP$ is \emph{not} compression oblivious for $\mathsf{M}$. This means there exists a quantum algorithm that computes $\propertyP$ on truth tables of functions corresponding to $\mathsf{M}$ in $N^{1-\alpha}$ time for some constant $\alpha>0$. Let $\mathcal{A}$ denote this algorithm. Without loss of generality, let us assume $\mathcal{A}$ is a $1/10$-bounded-error quantum algorithm. Let $\mathcal{B}$ denote the algorithm that computes $\propertyP$ in $\DTIME{N}$ as promised earlier. Using $\mathcal{A}$ and $\mathcal{B}$, we now construct an algorithm $\mathcal{D}$ that, on all inputs $z$, does the following:
\begin{enumerate}
    \item Run $\mathcal{A}$ on $z$ for $M=O(1)$ many times and note the outputs.
    \item Run $\mathcal{B}$ on $z$ once.
    \item \texttt{ACCEPT} if over $M$ runs of $\mathcal{A}$ on input $z$ the number of times $\mathcal{A}(z)=\mathcal{B}(z)$ is \emph{at most} $\frac{8}{10} \cdot M$, otherwise \texttt{REJECT}. (We need the success probability of $\mathcal{D}$ to be at least $2/3$ therefore choosing $M=O(1)$ suffices.)
\end{enumerate}
Let $z$ be a Boolean string of length $N=2^n$. Let $D_n$ denote the set of $N$-bit strings that $\mathcal{D}$ accepts with at least $2/3$ success probability. It is easy to see that $\mathcal{D}$ rejects strings for which $\mathcal{A}$ correctly computes $\propertyP$ (i.e., with success probability $9/10$). Recall that, correctness of $\mathcal{A}$ is only promised for strings that are truth tables of functions having $\poly(n)$-size circuits. There is no guarantee on how $\mathcal{A}$ behaves outside of this promise. But, we do know that $\propertyP$ has a quantum query complexity of $\Omega(N)$ on all sets of size at least $2^N \left(1-\frac{1}{N^{2}}\right)$. As $\mathcal{A}$ is a strictly-sublinear time algorithm (hence, also a strictly-sublinear query algorithm), it can at best correctly compute $\propertyP$ on at most $2^N \left(1-\frac{1}{N^{2}}\right)$ many strings outside of its promise. This, by construction, means $\mathcal{D}$ accepts at least $\frac{2^N}{N^2}$ many strings of length $N$ and therefore $D_n$ will contain these strings (or alternatively these functions). Now it is easy to see that $\mathcal{D}$ computes a natural property $(D_n)_{n \in \mathbb{N}}$, of density $(\delta_n=\frac{1}{2^{2n}})_{n \in \mathbb{N}}$, in $2\cdot 2^n$ quantum time. Moreover, this property $(D_n)_{n \in \mathbb{N}}$ is useful against complexity class $\mathsf{M}$.

Once we have the natural property which is useful against $\mathsf{M}$, we run a similar argument as done by Razborov and Rudich. As the algorithm $\mathcal{D}$ computes $(D_n)$, it can distinguish truth tables of functions with small circuits (i.e., circuits of $\polylog(n)$-depth and $\poly(n)$-size or simply $\poly(n)$-size) from truth tables of truly random functions with $\frac{1}{2^{2n}}$ advantage. But many believe that it is not possible to have quantum adversaries that are bounded by $2\cdot 2^n$ time and distinguish pseudorandom functions constructible in $\mathsf{M}$ with $\frac{1}{2^{2n}}$ advantage. Then it must be the case that $\propertyP$ is compression oblivious for $\mathsf{M}$.


\paragraph{Part 2: Fourier analysis of Boolean functions} As discussed above, our impossibility results hold for properties that are computable in $\DTIME{\linear}$ and that have high quantum query complexity even when the computation is restricted to strings from sets of size $2^N(1-\frac{1}{N^2})$. Both $\propertyParity$ and $\propertyMajority$ can be computed in $\DTIME{\linear}$. Moreover, using techniques from Fourier analysis, we are able to show that $\propertyParity$ and $\propertyMajority$ also require $\Omega(N)$ many quantum queries on sets of size $2^N(1-\frac{1}{N^2})$. Therefore, proving \Cref{thm:MainThmAsStatedInIntro}.
 
\subsection{Related work} 
\begin{itemize}
    \item Chen \textit{et al.}\ in \cite{CCKSS23QSETH} use the $\QSETH$ framework to conjecture quantum time lower bounds for variants of $\CNFSAT$ such as $\oplus\CNFSAT$ and $\#\CNFSAT$, etc and show lower bounds for several computational problems, such as variants of lattice problems, set-cover, hitting-set, etc., as implications.
    \item Arunachalam \textit{et al.}\ in \cite{AGGOS21} also discuss the notion of quantum natural properties but they do so to relate learning algorithms and circuit lower bounds. As one of the results, they show that containment of the Minimum Circuit Size Problem (MCSP) in \BQP{} implies new circuit lower bounds. This type of result was later extended by \cite{CCZZ22} to show that a few (more) quantum complexity classes can be separated if the Minimum Quantum Circuit Size Problem (MQCSP) is in \BQP{}. Despite the overlap in topic, both of these works turn out to be complementary to the current work. We use the notion of quantum natural properties to strengthen Buhrman \textit{et al.}'s \QSETH{} framework by grounding assumptions in cryptography.
    \item Contrary to Razborov and Rudich's impossibility result, Chow \cite{Cho11} proves the existence of a natural proof useful to imply circuit lower bounds as a consequence of the Exponential-Time Hypothesis (ETH). 
\end{itemize}

\subsection{Discussion and open questions}

In agreement with \cite{BPS21}, we also believe that the notion of \emph{compression-oblivious} properties is of independent interest. (See the blog post by Aviad Rubinstein also expressing interest \cite{Rub19Blog}.)
\begin{itemize}
    \item There may be other connections with topics in meta-complexity theory (as the one we witness in this paper) or say, connections with obfuscation as suspected by Aviad Rubinstein \cite{Rub19Blog} or connections with Kolmogorov complexity (as suggested by Harry Buhrman).
    \item Note that our results cannot be extended to show $\propertyParity$ and $\propertyMajority$ are compression oblivious for $\ACconstDepth$ circuits. This is because $\ACconstDepth$-constructible pseudorandom functions are not secure against $2^{\polylog(n)}$-time adversaries \cite{LMN93AC0Learnable}. Therefore, extending this result to $\propertyParity$ or $\propertyMajority$ for \acQSETH{}, i.e.\ $\CNF$ or $\DNF$ input, would be another step towards grounding the (necessary) assumption that such properties are compression oblivious.
    \item It is also worthwhile to study compression obliviousness of properties that don't meet the criteria discussed in our paper. There might be properties for which one may be able to construct natural proofs of slightly worse densities potentially implying some other complexity theoretic results as shown in \cite{Cho11}.
\end{itemize}

On the other front, our proof technique also allows us to, in a rather black-box way, check if a particular property is likely to be compression-oblivious. One could plausibly extend our results to a larger set of functions that have a similar structure, e.g., a natural candidate would be to show an equivalent statement for symmetric functions with non-negligible mass on high-degree Fourier coefficients. It is very likely that there are in fact many properties of the type we discuss in this paper; one could use techniques similar to ones used in \cite{Ambainis99-QueryCompForAlmostAllF} to formally estimate the number of such properties.

\subsection{Acknowledgements}
We are grateful to Harry Buhrman for many useful discussions, including those surrounding his original conception of the notion of compression oblivious properties, which already envisioned the type of connection to meta-complexity theory we present in this work.
We are also grateful to Ronald de Wolf for useful discussions and pointing out an ambiguity in an earlier draft.

\subsection{Structure of the paper}
In \Cref{sec:Preliminaries}, we present useful definitions and notations. In \Cref{sec:QSETHframework}, we formally define the $\QSETH$ framework. Our definitions omit an ambiguity from its original version presented in \cite{BPS21}. In \Cref{sec:COproperties}, we prove our main result - the complexity-theoretic evidence for compression obliviousness of $\propertyParity$ and $\propertyMajority$. A few proofs useful for this section are moved to \Cref{AppSec:ProofsRelevantToSupportingCOassumption}. Additionally, in \Cref{Appsec:OldObservationsAboutCO} we reprove some of the earlier results about compression-oblivious properties that were given in \cite{BPS21}; we do this to convince the readers of the correctness of our proposed definition of the $\QSETH$ framework.

\section{Preliminaries}
\label{sec:Preliminaries}

Most of the definitions and notions presented in this section are taken from \cite{AB09}, and we restate them here for completeness and ease of reading.

\subsection{Uniform models of computation}
\label{sec:UniformModel}

\begin{defn}[Function and language recognition]
\label{def:FunctionNlanguageRecognition}
Let $f: \{0,1\}^{*} \rightarrow \{0,1\}$ be a Boolean function. The language corresponding to $f$ is the set of strings $L_f$ on which $f$ evaluates to $1$, i.e., $\forall x \in \{0,1\}^*$,
\begin{equation*}
    x \in L_f \iff f(x)=1.
\end{equation*}
Additionally, we say a Turing machine $M$ \emph{decides} a language $L_f$ if for all $x \in \{0,1\}^*$, whenever $M$ is initialized to the start configuration on input $x$, the machine $M$ halts with $f(x)$ written on its output tape.
\end{defn}

\begin{defn}[$\textsc{DTIME}$, Definition~1.12 in \cite{AB09}]
Let $T: \mathbb{N} \rightarrow \mathbb{N}$ be a function. A language $L$ is in $\DTIME{T(n)}$ if there exists a Turing machine $M$, and there exists a constant $c$, such that $\forall n, \forall x \in \{0,1\}^n$, $M$ on input $x$ runs in $c\cdot T(n)$ steps and decides $L$. 
\end{defn} 

\begin{defn}[$\textsc{BQTIME}$, adapting Definition 10.9 in~\cite{AB09}]
Let $f: \{0,1\}^*\rightarrow \{0,1\}$ and $T: \mathbb{N} \rightarrow \mathbb{N}$ be a function. We say $f$ is computable in quantum $T(n)$-time, or (equivalently) say $L_f\in \BQTIME{T(n)}$, if there is a classical polynomial-time Turing machine that, for every $n\in\mathbb{N}$, on input $(1^n,1^{T(n)})$ outputs the descriptions of a $T(n)$-size quantum circuit,\footnote{\label{footnote:ckt_description}A few clarifications.\begin{itemize}
      \item A description of a $T$-size quantum circuit $F\coloneqq (F_1,\ldots,F_{T})$ is an ordered sequence of $T$ elements, where $\forall i \in [T]$ each element $F_i=(G_i,w_i)$ specifies a quantum elementary gate $G_i \in \{\texttt{CNOT}\} \cup \{~\text{set of all $1$-qubit operations}~\}$ and the label(s) of the wire(s) $w_i$ on which gate $G_i$ will be applied.
      \item The action ``applying an elementary quantum operation $F_i=(G_i,w_i)$'' means applying the elementary gate $G_i$ on wire specified by location(s) $w_i$.
      \item The size of the quantum circuit, on the other hand, is the number of elements in its description.
  \end{itemize}} let us denote it by $F \coloneqq (F_1,\ldots,F_{T(n)})$, such that for every $x\in\{0,1\}^n$, the output of the following process is $f(x)$ with probability at least $2/3$:
\begin{enumerate}
    \item Initialize an $m$-qubit quantum state $\ket{x}\ket{0^{m-n}}$, where $m\leq T(n)$.
    \item Apply one after the other $T(n)$ elementary quantum operations $F_1,\ldots,F_{T(n)}$ to the state.\footnote{See \cref{footnote:ckt_description}.}
    \item Measure the last qubit in the computational basis and output the measurement outcome. 
\end{enumerate}
\end{defn}
A language $L_f$ corresponding to a Boolean function $f:\{0,1\}^*\rightarrow \{0,1\}$ is in BQP if there exists a polynomial $p$ such that $f$ is computable in quantum $p(n)$-time, or equivalently, $L_f\in\BQTIME{p(n)}$.

\subsection{Non-uniform model of computation - Boolean circuits and circuit families}
\label{sec:BooleanCircuitsNcircuitFamily}

\begin{defn}[Boolean circuits, Definition~6.1 in \cite{AB09}]
\label{sec:BooleanCircuits}
For every $n \in \mathbb{N}$, an $n$-input, single-output Boolean circuit is a directed acyclic graph with $n$ sources (i.e., vertices with no incoming edges but could have an unbounded number of outgoing edges) and one sink (i.e., vertex with no outgoing edges); let us denote the sink vertex by $s$. All non-source vertices are called gates and are labeled with one of $\land$, $\lor$ or $\neg$ (the logical operations AND, OR and NOT, respectively). The vertices labeled with $\land$ and $\lor$ have fan-in equal to $2$ and those labeled with $\neg$ gate have fan-in $1$. The source vertices are labeled with indices $i \in [n]$. We use $\texttt{label}(v)$ to denote the label of a vertex~$v$. The circuit-size of $\circuit{C}$, denoted by $|\circuit{C}|$, is the total number of vertices in it.

If $\circuit{C}$ is a Boolean circuit and $x \in \{0,1\}^n$ is some input, then the output of $\circuit{C}$ on $x$, denoted by $\circuit{C}(x)$, is defined in the natural way. More formally, for every vertex $v$ of $\circuit{C}$ we define a $\texttt{val}(v)$ as follows: if $\texttt{label}(v)=i$ then $\texttt{val}(v)=x_i$ and otherwise $\texttt{val}(v)$ is defined recursively by applying the logical operation of $v$ labeled by $\texttt{label}(v)$ on the values of vertices connected to $v$. The output $\circuit{C}(x)$ is the value of the sink vertex $\texttt{val}(s)$.
\end{defn}

In any situation, if we require that the Boolean circuits are only allowed to have $\land$ and $\lor$ gates with fan-in $2$, as opposed to these gates having unbounded fan-in, then we explicitly mention it.

\begin{defn}[Circuit family and language recognition, Definition~6.2 in \cite{AB09}]
\label{def:CircuitFamily}
Let $T: \mathbb{N} \rightarrow \mathbb{N}$ be a function. A \emph{$T(n)$-size circuit family} is a sequence $\{C_n\}_{n \in \mathbb{N}}$ of Boolean circuits where for every $n$, $C_n$ has $n$ inputs and a single output and its circuit size $|C_n| \leq T(n)$.
\end{defn}

\begin{defn}[\textsf{SIZE}]
Let $T: \mathbb{N} \rightarrow \mathbb{N}$ be a function. We say a language $L \in \Size{T}$ if there exists a $T(n)$-size circuit family $\{C_n\}_{n \in \mathbb{N}}$ such that $\forall n$, $C_n$ consists of $\land$ and $\lor$ gates with fan-in $2$ and $\neg$ gates of fan-in $1$ and, moreover, $\forall x \in \{0,1\}^n, x \in L \iff C_n(x)=1$.
\end{defn}

\begin{defn}[Complexity class]
\label{def:ComplexityClass} 
A complexity class is a set of languages. 
\end{defn}

\begin{defn}[Circuit family corresponding to a language] Let $L$ be a language (as in \cref{def:FunctionNlanguageRecognition}). We say a family of circuits $\{C_n\}_{n \in \mathbb{N}}$ corresponds to language $L$, if $\forall n$, $\forall x \in \{0,1\}^n$, we have $x \in L \iff C_n(x)=1$. 
\end{defn}

\begin{defn}[$\Ppoly$, Definition~6.5 in \cite{AB09}]
    $\Ppoly$ is the class of languages that are decidable by polynomial-sized circuit families, or equivalently, $\Ppoly=\cup_{c,d}\Size{c \cdot n^d}=\cup_{c,d}\SizeArg{c}{d}$.
\end{defn}


\begin{defn}[$\AC$, Adapting Definition~6.25 in \cite{AB09}]
\label{def:AC0circuits}
A language $L$ is in $\AC$ if there exists a polynomial $p:\mathbb{N} \rightarrow \mathbb{N}$ and there exists a $p(n)$-size circuit family $\{C_n\}_{n \in \mathbb{N}}$ corresponding to $L$ where $\forall n$, circuit $C_n$ uses $\neg$, $\land$ and $\lor$ gates, with latter two allowed to have unbounded fan-in, and has depth at most $2$.

\end{defn}


\begin{defn}[$\TypeO$ complexity class]
\label{def:SizeableComplexityClass} 
We say a complexity class is of $\TypeO$ if it is any of the following complexity classes:
\begin{enumerate}
    \item For every $d \in \mathbb{N}$ and $k \in \mathbb{N}$, we include the class $\ACarg{d}{k}$; we say a language $L \in \ACarg{d}{k}$ if there exists a family of Boolean circuits $\{C_n\}_{n \in \mathbb{N}}$ corresponding to $L$ such that $\forall n$, $C_n$ uses $\neg$, $\land$ and $\lor$ gates, with latter two allowed to have unbounded fan-in, and has depth at most $d$ and circuit size $|C_n| \leq n^k$. 
    \item For every $d \in \mathbb{N}$ and $k \in \mathbb{N}$, we include the class $\NCarg{d}{k}$; we say a language $L \in \NCarg{d}{k}$ if there exists a family of Boolean circuits $\{C_n\}_{n \in \mathbb{N}}$ corresponding to $L$ such that $\forall n$, $C_n$ has depth at most $\log^d n$ and circuit size $|C_n| \leq n^k$. 
    \item For every $d \in \mathbb{N}$ and $k \in \mathbb{N}$, we include the class $\SizeArg{d}{k}$; we say a language $L \in \SizeArg{d}{k}$ if there exists a family of Boolean circuits $\{C_n\}_{n \in \mathbb{N}}$ corresponding to $L$ such that $\forall n$, $C_n$ has circuit size $|C_n| \leq d \cdot n^k$.
\end{enumerate}
\end{defn}

\subsection{Truth tables of circuits, languages and a few other useful notations}

A truth table is binary string capturing the outputs of a Boolean function. If a function is defined on $n$ binary inputs then its corresponding truth table is a Boolean string in $\{0,1\}^{2^n}$. The notion of a truth table can also be extended to capture the membership of strings in a language by fixing the length of the strings. We formally define these notions as follows.

\begin{defn}[Truth tables] 
\label{def:Truthtables}
For any $n \in \mathbb{N}$, let $f_n : \{0,1\}^n \rightarrow \{0,1\}$ be a Boolean function. Then the truth table of $f_n$, denoted by $\truthtable(f_n)$, is a $2^n$-length binary string constructed in the following way:
\begin{equation*}
    \truthtable(f_n)=\concatenate_{a \in \{0,1\}^n} f_n(a).
\end{equation*}
Here $f_n(a)$ denotes the output $f_n$ on input $a$. And for the concatenation, the strings $a \in \{0,1\}^n$ are taken in lexicographic order.

For a language $L \subseteq \{0,1\}^*$ the notion of truth table is as follows: for any $n \in \mathbb{N}$, the truth table of $L$ at $n$, denoted by $\truthtable(L, n)$, is a $2^n$-length binary vector that can be obtained by concatenating the outputs $b_a \coloneqq [a \in L]$ for all $a \in \{0,1\}^n$ taken in lexicographic order.  
\end{defn}

\paragraph{Other useful notations.}
Let $n$ be an integer. Let $f_n:\{0,1\}^n \rightarrow \{0,1\}$ be a Boolean function on $n$ input variables. 
\begin{itemize}
    \item We use $\setOfAllStrings_n$ to denote the set of \textit{all} Boolean functions defined on $n$ input variables; We represent $\setOfAllStrings_n = \{ \truthtable(f_n) \mid f_n \text{ is a Boolean function on }n \text{ variables} \}$ as the set of truth tables of all functions $f_n:\{0,1\}^n \rightarrow \{0,1\}$, implying $\setOfAllStrings_n = \{0,1\}^{2^n}$ and $|\setOfAllStrings_n|=2^{2^n}$. We also use $N=2^n$.
    \item We overload the notation of $\truthtable(\cdot)$ on sets of Boolean functions to denote the following: if $S$ is a set of Boolean functions on $n$ variables then $\truthtable(S)= \{\truthtable(f_n) \mid f_n \in S \}$.
    \item Let $\Lambda$ be a complexity class. Let $L \in \Lambda$. Furthermore, let $f_n:{0,1}^n \rightarrow \{0,1\}$ be a Boolean function such that $\truthtable(f_n)=\truthtable(L,n)$ then we say $f_n$ \emph{is a sub-element of} $\Lambda$.
    \item We sometimes also overload the notation of $\truthtable(\cdot)$ on complexity classes to denote the following: if $\Lambda$ is a complexity class then $\truthtable(\Lambda,n)= \{\truthtable(f_n) \mid f_n \textnormal{ sub-element in } \Lambda \}$ or equivalently $\truthtable(\Lambda,n)= \{\truthtable(L,n) \mid L \in \Lambda \}$. Furthermore, we use $\truthtable(\Lambda)=\bigcup_{n \in \mathbb{N}} \truthtable(\Lambda,n)$.
\end{itemize}

\section{The QSETH framework}
\label{sec:QSETHframework}

\subsection{Definitions}

\begin{defn}[Property]
\label{def:BooleanProperty}
A Boolean property (or just property) is a sequence $\propertyP \coloneqq\big( \propertyP_n\big)_{n\in\mathbb{N}}$ where each $\propertyP_n$ is a set of Boolean functions defined on $n$ variables.
\end{defn}

Alternatively, one can view $\propertyP_n$ as a set of $2^n$-length Boolean strings that are truth tables of functions $f_n \in \propertyP_n$, or equivalently view it as a function $\propertyP_n: \{0,1\}^{2^n} \rightarrow \{0,1\}$ such that $\propertyP_n(\truthtable(f_n))=1$ if and only if $f_n \in \propertyP_n$; see \Cref{def:Truthtables} for the definition of a truth table. Also note that it is possible for a property $\propertyP$ to be \emph{partial}. In such a case, for every $n$, there could be some functions $f_n$ for which we don't care whether or not they belong to $\propertyP_n$.\footnote{For example, both the properties $\mathrm{PP}_{\mathrm{edit}}$ and $\mathrm{PP}_{\mathrm{lcs}}$ considered in \cite{BPS21} are partial.}

\paragraph{Quantum query complexity of $\propertyP$} The (bounded-error) quantum query complexity is defined only in a non-uniform setting; therefore, we define this notion for $\propertyP_n$ for every $n \in \mathbb{N}$.  A quantum query
algorithm $\mathcal{B}$ for $\propertyP_n: \{0,1\}^{2^n} \rightarrow \{0,1\}$, on an input $x \in \{0,1\}^{2^n}$ begins in a fixed initial state $\ket{\psi_0}$, applies a sequence of unitaries $U_0, O_x, U_1, O_x, \ldots, U_T$, and performs a measurement whose outcome is denoted by $z$. Here, the initial state $\ket{\psi_0}$ and the unitaries $U_0, U_1, \ldots, U_T$ are independent of the input $x$. The unitary $O_x$ represents the ``query'' operation, and maps $\ket{i}\ket{b}$ to $\ket{i}\ket{b + x_i \textnormal{ mod } 2}$ for all $i \in [2^n]-1$. 
We say that $\mathcal{B}$ is a $1/3$-bounded-error algorithm computing $\propertyP_n$ if for all $x$ in the domain of $\propertyP_n$, the success probability of outputting $z=\propertyP_n(x)$ is at least $2/3$. Let cost$(\mathcal{B})$ denote the number of queries $\mathcal{B}$ makes to $O_x$ throughout the algorithm. The $1/3$-bounded-error quantum query complexity of $\propertyP_n$, denoted by $\Q_{1/3}(\propertyP_n)$, is defined as 
$\Q_{1/3}(\propertyP_n)=\min\{\text{cost}(\mathcal{B}):\mathcal{B} \text{ computes } \propertyP_n \text{ with error probability } \leq 1/3\}$.\footnote{Note that if $\propertyParity \coloneqq (\propertyParity_n)_{n \in \mathbb{N}}$ is the property we are considering, under our definition of $\propertyParity_n$ the value of $\Q(\propertyParity_{n})$ is equal to $\Omega(2^{n})$ and not $\Omega(n)$. For the property $\propertyOR \coloneqq (\propertyOR_n)_{n \in \mathbb{N}}$ then $Q(\propertyOR_{n})=\Omega(2^{\frac{n}{2}})$.}

\paragraph{Algorithms computing a property $\propertyP$ in black-/white-box setting} We say a (bounded-error quantum) algorithm $\mathcal{A}$ computes a property $\propertyP$ (as in \Cref{def:BooleanProperty}) in the \emph{black-box} setting if for all $n\in \mathbb{N}$, for all $f_n$ given its truth table as input, i.e., a $2^n$-length Boolean string, $\mathcal{A}$ makes (quantum) queries to this input and (with success probability at least $2/3$) determines whether or not $f_n \in \propertyP_n$. An algorithm computing $\propertyP$ in the \emph{white-box} setting gets a circuit description of $f_n$, which is possibly a more succinct description of $f_n$'s truth table, and the algorithm has to determine whether or not $f_n \in \propertyP_n$. We emphasize that, in both the black-/white-box setting, we care about the \emph{time} that algorithm $\mathcal{A}$ takes when deciding whether or not $f_n \in \propertyP_n$.
\newline

We now present the definition of compression-oblivious properties. In simple terms, these are the properties whose quantum query complexity is a lower bound for the time complexity to compute the property even for inputs that are ``compressible'' as a truth table of small formulas (or circuits). Formally,

\begin{defn}[$\Gamma_{d,k}$-compression-oblivious properties]
\label{def:CompressionOblivious}
Let $d,k \in \mathbb{N}$. Let $\Gamma_{d,k}$ denote a $\TypeO$ complexity class as stated in \cref{def:SizeableComplexityClass} parameterized by $d$ and $k$. We say a property $\propertyP$ is \emph{$\Gamma_{d,k}$-compression-oblivious}, denoted by $\propertyP \in \mathcal{CO}(\Gamma_{d,k})$, if for every constant $\delta>0$, for every quantum algorithm $\mathcal{A}$ that computes $\propertyP$ in the \textbf{black-box} setting, $\forall n' \in \mathbb{N}, \exists n \geq n'$ and $\exists$ a set $L=\{L^1, L^2, \ldots\} \subseteq \Gamma_{d,k}$ of `hard languages', such that $\forall$ circuit families $\{C^1_{n''}\}_{n'' \in \mathbb{N}}$ corresponding to $L^1$, $\forall$ circuit families $\{C^2_{n''}\}_{n'' \in \mathbb{N}}$ corresponding to $L^2$, $\ldots$, the algorithm $\mathcal{A}$ uses at least $\Q_{1/3}(\propertyP_{n})^{1-\delta}$ quantum time on at least one of the inputs in $\{C^i_n\}_{i \in [|L|]}$.\footnote{While the idea of compression-oblivious properties can be easily expressed informally, describing the notion formally requires discussing intricate details to bridge the connections between uniform (e.g., algorithm $\mathcal{A}$ computing $\propertyP$) and non-uniform models of computations (e.g., query complexity $\Q(\propertyP_n)$).}
\end{defn}


\begin{defn}[compression-oblivious properties w.r.t.\  $\AC$, $\NC$, $\Ppoly$]
\label{def:COforGeneralComplexityClass}
We say a property $\propertyP : \{0,1\}^* \rightarrow \{0,1\}$ is
\begin{enumerate}
    \item $\propertyP \in \mathcal{CO}(\AC)$ if $d=2$ and $\exists k \in \mathbb{N}$ such that $\propertyP \in \mathcal{CO}(\ACarg{d}{k})$, 
    \item $\propertyP \in \mathcal{CO}(\ACconstDepth)$ if $\exists d \in \mathbb{N}, \exists k \in \mathbb{N}$ such that $\propertyP \in \mathcal{CO}(\ACarg{d}{k})$, 
    \item $\propertyP \in \mathcal{CO}(\NC)$ if $\exists d \in \mathbb{N}, \exists k \in \mathbb{N}$ such that $\propertyP \in \mathcal{CO}(\NCarg{d}{k})$, and
    \item $\propertyP \in \mathcal{CO}(\Ppoly)$ if $\exists d \in \mathbb{N}, \exists k \in \mathbb{N}$ such that $\propertyP \in \mathcal{CO}(\SizeArg{d}{k})$.
\end{enumerate}
    
\end{defn}

The original definition of compression-oblivious properties, as presented in Section~2.2 of \cite{BPS21}, does not clearly enforce a unique ordering of the quantifiers, which can potentially lead to undesirable consequences. Hence, we make the effort of presenting a more formal definition of $\Gamma$-compression-oblivious properties in  \Cref{def:CompressionOblivious,def:COforGeneralComplexityClass} (with the approval of the authors \cite{BPS21}). 

Having formally defined $\Gamma$-compression-oblivious properties for all $\Gamma \in \{\AC, \ACconstDepth, \NC, \Ppoly \}$ we now state the $\Gamma$-QSETH conjecture. Earlier in this section, we discussed the difference between computing a Boolean property in a black-box setting versus computing the same property in a white-box setting. The compression-oblivious properties are those properties for which its corresponding quantum query complexity lower bounds the time complexity in the \emph{black-box} setting. For such properties, Buhrman \textit{et al.}\ conjectured that the time complexity of computing these properties in the \emph{white-box} setting is as bad as it is in the black-box setting. This assumes that the succinct description of the Boolean functions doesn't help in computing the properties more efficiently. (This is similar to the \textit{Black-Box Hypothesis (BBH)} stated in \cite{IKKMR17CircuitHelp}.) More formally,

\begin{conjecture}[$\Gamma$-$\QSETH$, a more formal presentation of Conjecture~5 in \cite{BPS21}]
Let $\Gamma \in \{\AC, \ACconstDepth, \NC, \Ppoly \}$ be a complexity class (as in \cref{def:ComplexityClass}). Let $\propertyP$ be a \emph{$\Gamma$-compression-oblivious} property. The $\Gamma$-$\QSETH$ conjecture states that $\exists d, \exists k \in \mathbb{N}$, such that for every constant $\delta>0$, for every quantum algorithm $\mathcal{A}$ that computes $\propertyP$ in the \textbf{white-box} setting, $\forall n' \in \mathbb{N},\exists n \geq n'$, $\exists$ a set $L=\{L^1, L^2, \ldots\} \subseteq \Gamma_{d,k}$, $\forall $ circuit families $\{C^1_{n''}\}_{n'' \in \mathbb{N}}$ corresponding to $L^1$, $\forall$ circuit families $\{C^2_{n''}\}_{n'' \in \mathbb{N}}$ corresponding to $L^2$, $\ldots$, the algorithm $\mathcal{A}$ uses at least $\Q_{1/3}(\propertyP_{n})^{1-\delta}$ quantum time on at least one of the inputs in $\{C^i_n\}_{i \in [|L|]}$. Note that $\Gamma_{d,k}$ is a $\TypeO$ complexity class corresponding to $\Gamma$.
\end{conjecture}

\subsection{Relevance}

Note that a basic version of the $\QSETH$ conjecture and its implications has already been discussed in complementary works \cite{BasicQSETH-Aaronson-2020, BPS21}, and other variants of $\QSETH$ are discussed more recently in \cite{CCKSS23QSETH} and in Chen's PhD thesis \cite{Che24Thesis}. For example,

\begin{conjecture}[Basic-$\QSETH$, Conjecture~8.3 in \cite{Che24Thesis}]
\label{conj:BasicQSETHfromYanlinThesis}
For each constant $\delta >0$, there exists $c>0$ such that there is no bounded-error quantum algorithm that solves $\CNFSAT$ (even restricted to formulas with at most $cn^2$ many clauses) in $O(2^{\frac{n(1-\delta)}{2}})$ time.    
\end{conjecture}

\begin{conjecture}[$\oplus\QSETH$, Conjecture~8.9 in \cite{Che24Thesis}]
\label{conj:ParityQSETHfromYanlinThesis}
For each constant $\delta >0$, there exists $c>0$ such that there is no bounded-error quantum algorithm that solves $\oplus\CNFSAT$ (even restricted to formulas with at most $cn^2$ many clauses) in $O(2^{n(1-\delta)})$ time.    
\end{conjecture}

\begin{conjecture}[$\#\QSETH$, Conjecture~8.8 in \cite{Che24Thesis}]
\label{conj:CountingQSETHfromYanlinThesis}
For each constant $\delta >0$, there exists $c>0$ such that there is no bounded-error quantum algorithm that solves $\#\CNFSAT$ (even restricted to formulas with at most $cn^2$ many clauses) in $O(2^{n(1-\delta)})$ time.    
\end{conjecture}

We are able to get these conjectures as a corollary to our QSETH framework. More precisely,

\begin{corollary}
\label{cor:AC2QSETHimpliesConjectureInYanlinThesis}
The $\AC$-$\QSETH$ implies \Cref{conj:BasicQSETHfromYanlinThesis}. Additionally, if we assume that $\propertyParity \in \mathcal{CO}(\AC)$ then $\AC$-$\QSETH$ implies \Cref{conj:ParityQSETHfromYanlinThesis,conj:CountingQSETHfromYanlinThesis}.
\end{corollary}
\begin{proof}
The proof of $\propertyOR \in \mathcal{CO}(\AC)$ was first given in \cite{BPS21} (we also present the proof in \Cref{Appsec:OldObservationsAboutCO}). Thereafter, setting $d=2$ and $k=3$ for $\AC$-$\QSETH$ conjecture implies \Cref{conj:BasicQSETHfromYanlinThesis,conj:ParityQSETHfromYanlinThesis,conj:CountingQSETHfromYanlinThesis}.
\end{proof}

Therefore, all the $\QSETH$-based lower bounds presented in Chen's PhD thesis \cite{Che24Thesis} immediately follow from our $\QSETH$ framework.

\section{When certain properties are not $\Ppoly$-compression-oblivious}
\label{sec:COproperties}


In this section, we show that it is reasonable to conjecture that there exist $d',k' \in \mathbb{N}$ and $d'',k'' \in \mathbb{N}$ for which properties such as $\propertyParity$ and $\propertyMajority$ are $\SizeArg{d'}{k'}$-compression-oblivious and $\SizeArg{d''}{k''}$-compression-oblivious, respectively. We do so by relating this conjecture to topics in circuit complexity and pseudorandomness. In fact, we are able to show similar consequences for properties that satisfy certain constraints. $\propertyParity$ and $\propertyMajority$ happen to be such properties.

Following the notation about natural proofs from \cite{RR94} and \cite{Cho11}, we now define the following.

\begin{defn}[$\Upsilon$-natural property with density $\delta$]
\label{defn:UpsilonNaturalProperty}
Let $\Upsilon$ be a complexity class (as in \cref{def:ComplexityClass}) and let $\delta \coloneqq (\delta_n)_{n \in \mathbb{N}}$ denote a sequence of positive real numbers. A property $D\coloneqq (D_n)_{n \in \mathbb{N}}$ (as in \cref{def:BooleanProperty}) is $\Upsilon$-natural with density $\delta \coloneqq (\delta_n)_{n \in \mathbb{N}}$ if the following statements hold:
\begin{enumerate}
    \item largeness: $\exists n', \forall n>n'$, $|D_n| \geq 2^{2^n} \cdot \delta_n$, and
    \item constructivity: the problem of determining whether a Boolean function $f_n \in D_n$, when given as input the truth table of $f_n$ on $n$ variables, is decidable in $\Upsilon$, which in more formal terms translates to saying $L_{D} \in \Upsilon$ where $L_{D}\coloneqq \bigcup_{n \in \mathbb{N}} \{\truthtable(f_n) \mid f_n \in D_n \}$.
\end{enumerate}  
\end{defn}

\begin{defn}[Quasi-useful property, Definition~5 in \cite{Cho11}]\footnote{Note that Razborov and Rudich use the notion of \emph{useful} property which is defined for infinitely many $n$. However, for our results the notion of \emph{quasi-useful} property (defined for all sufficiently large $n$) makes more sense so we use this instead.}
\label{defn:Quasi-UsefulProperty}
Let $\Lambda$ denotes a complexity class. A property $D \coloneqq (D_n)$, or equivalently the corresponding language $L_{D}\coloneqq \bigcup_{n \in \mathbb{N}} \{\truthtable(g_n) \mid g_n \in D_n \}$, is quasi-useful against complexity class $\Lambda$ if $\forall (f_n)$, we have $(f_n) \notin \Lambda$ whenever $(f_n)$ is sequence of Boolean functions satisfying $f_n \in D_n$ for all sufficiently large $n$.
\end{defn}

Our goal is to realize the chain of arguments as illustrated in \Cref{fig:flowsforCO}.

\begin{figure}[h]
\centering
\begin{tikzpicture}[>={Classical TikZ Rightarrow[]}]
\node[fill=white,
      thick,
      align=left,
      draw,
      minimum height=2cm,
      minimum width=2.5cm,
      label=south: $^*$for certain $\propertyP$s 
] at (0,0) (B) {P $\notin \mathcal{CO}(\Lambda)$} ;

\node[fill=white,
      thick,
      align=left,
      draw,
      minimum height=2cm,
      minimum width=3.5cm,
] at (6.3,0) (A) {$\exists \text{ } \BQTIME{2 \cdot 2^{n}}\text{-natural }$ \\ $ D_n$ with density 
$\delta_n=\frac{1}{2^{2n}}$\\$ \text{quasi-useful against } \Lambda $};
\node[fill=white,
      thick,
      align=left,
      draw,
      minimum height=2cm,
      minimum width=3.5cm,
] at (13,0) (C) {$\nexists \text{ quantum-secure } $ \\ $\Lambda\text{-constructible PRFs}$};
\node[fill=white,
      thick,
      align=left,
      draw,
      minimum height=2cm,
      minimum width=2.5cm, 
] at (2,-3.5) (D) { $\propertyMajority \text{ or } \propertyParity\notin \mathcal{CO}(\Lambda)$} ;
\node[fill=white,
      thick,
      align=left,
      draw,
      minimum height=2cm,
      minimum width=2.5cm, 
] at (10,-3.5) (E) {$\nexists \text{ quantum-secure } $ \\ $\Lambda\text{-constructible PRFs}$} ;
\node at (6.3,-2.5) (F){};
\draw[-latex] (B.east)-- (2.5,0) |- (A.west) node[midway, below] {\Cref{thm:LinearTimePropertyNotinFixedpCOimpliesQuasiUsefulProperty}};
\draw[-latex] (B.east)-- (2.5,0) |- (A.west) node[midway, above] {(A)};
\draw[-latex] (A.east)-- (10,0) |- (C.west)node[midway,below] {\cref{thm:QuasiusefulAgainstPpolyImpliesNoQsecurePRFs}};
\draw[-latex] (A.east)-- (10,0) |- (C.west)node[midway,above] {(B)};
\draw[-latex] (D.east)-- (6,-3.5) |- (E.west)node[midway,above] {};
\draw[double,->,double distance=2pt,=> angle 90] (6.3,-1.4)-- (6.3,-1.9) -| (F)node[midway,left] {(C)};
\draw[double,->,double distance=2pt,=> angle 90] (6.3,-1.4)-- (6.3,-1.9) -| (F)node[midway,right] {\cref{thm:ParityMajorityviolateQPRFassumptionIfNotCO}};
\end{tikzpicture}
\caption{Chain of arguments presented in this section. We discuss the implications when a certain class of properties $\propertyP$ are not $\Lambda$-compression-oblivious in links (A) and (B). In link (C) we show that $\propertyMajority$ and $\propertyParity$ belong to this class of properties.} 
\label{fig:flowsforCO}
\end{figure}
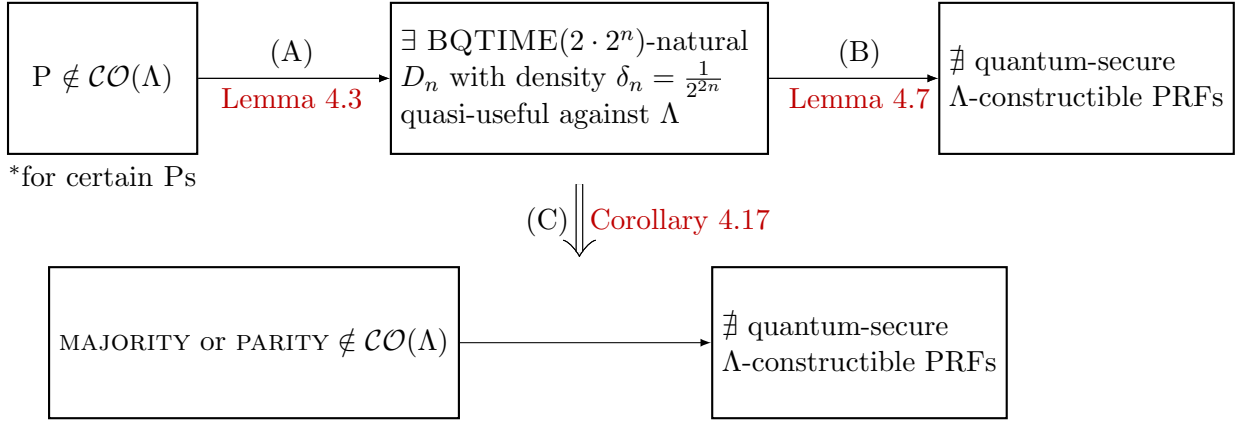

\paragraph{Realising link (A)} The proof idea is as follows. Let $\Lambda=\SizeArg{d}{k}$. Our choice of $d, k$ will be clear later. If $\propertyParity \notin \mathcal{CO}(\Lambda)$ (or say $\propertyMajority \notin \mathcal{CO}(\Lambda)$ may be for another choice of $d, k$) then (from the definition of compression-oblivious properties) there exists a quantum algorithm that computes the $\propertyParity$ (or $\propertyMajority$) of truth tables of circuits $C$ that are sub-elements of $\Lambda$ in time sublinear in the size of the input $\truthtable(C)$. Using this algorithm and the linear time deterministic algorithm to compute $\propertyParity$ (or $\propertyMajority$) on all strings of length $|\truthtable(C)|$, we will construct a $\BQTIME{\linear}$-natural property that is quasi-useful against complexity class $\Lambda$. Interestingly enough, the idea also generalizes to other simple properties that satisfy certain requirements; we formalize that in the statement of \cref{thm:LinearTimePropertyNotinFixedpCOimpliesQuasiUsefulProperty}.

\begin{lemma}
\label{thm:LinearTimePropertyNotinFixedpCOimpliesQuasiUsefulProperty}
Let $d,k \in \mathbb{N}$ and $\Lambda=\SizeArg{d}{k}$. If a property $\propertyP:\{0,1\}^* \rightarrow \{0,1\}$, or alternatively denoted as $(\propertyP_n)_{n \in \mathbb{N}}$ (as in \cref{def:BooleanProperty}), satisfies the following conditions:
\begin{enumerate}
    \item \label{item:1} $\propertyP \notin \mathcal{CO}(\Lambda)$,
    \item \label{item:2} $\propertyP$ on inputs of length $N$ is computable in $\DTIME{N}$, which more formally means $L_{\propertyP}\coloneqq \bigcup_{n \in \mathbb{N}} \{\truthtable(f_n) \mid f_n \in \propertyP_n \} \in \DTIME{N}$, and
    \item \label{item:3} $\forall S \subseteq \{0,1\}^N$ of size $|S|\geq 2^N(1 - \frac{1}{N^2})$, any $1/3$-bounded-error quantum query algorithm that computes $\propertyP_{\log N}$ on each $x \in S$ requires $\Omega(N)$ queries,
\end{enumerate}
 then there exists a $\BQTIME{2 \cdot N}$-natural property $(\propertyD_n)$ with density $\delta_n = \frac{1}{N^2}$ and this natural property is quasi-useful against complexity class $\Lambda$.
\end{lemma}
\begin{proof}
Let $d,k \in \mathbb{N}$ and let $\Lambda=\SizeArg{d}{k}$. From \cref{def:CompressionOblivious} we say, $\propertyP \notin \mathcal{CO}(\Lambda)$ if and only if $\exists \delta>0$, $\exists$ algorithm $\mathcal{A}$ computing property $\propertyP$ such that $\exists n' \in \mathbb{N}$ and $\forall n \geq n'$, $\forall L=\{L^1, L^2, \ldots \} \subseteq \Lambda$, $\mathcal{A}$ uses at most $Q_{1/3}(\propertyP_n)^{1-\delta}$ quantum time on all inputs in $\{\truthtable(L^i,n)\}_{i\in [|L|]}$.

Then, let $n \geq n'$, let $N=2^n$ and let $L$ be the largest subset of $\Lambda$, i.e., the set $\Lambda$ itself, $\mathcal{A}$ computes the property in at most $Q_{1/3}(\propertyP_n)^{1-\delta}$ time for all binary strings in $\truthtable(\Lambda)=\{\truthtable(L^i, n)\}_{i \in [|L|]}$. Furthermore, we have $|\truthtable(\Lambda)| \leq 2^{d' \cdot n^k}$ for some $d'>d$; this is because all the languages $L^i \in \Lambda$ which means for every $L^i$ there must exist a circuit family $\{C^i_n\}_{n \in \mathbb{N}}$ with $|C_n| \leq d\cdot n^k$ and the number of such circuits is at most $2^{d' \cdot n^k}$ which means the number of \emph{unique} $2^n$-length truth tables that $\Lambda$ can hold is at most $2^{d' \cdot n^k}$. W.l.o.g., let us assume that $\mathcal{A}$ computes $\propertyP$ with success probability at least $9/10$ (as we can boost the probability of success by running the $1/3$-bounded-error algorithm $\mathcal{A}$ and then take majority voting). For all $n > n'$, for all $x \in \truthtable(\Lambda,n)$, the algorithm $\mathcal{A}$ computes $\propertyP(x)$ in $\Q_{1/10}(\propertyP_{\log |x|})^{1-\delta}\leq |x|^{1-\delta}$ time for some constant $\delta>0$. 

In addition, let $\mathcal{B}$ denote a deterministic (quantum or classical) algorithm that computes $\propertyP$ on all $y \in \{0,1\}^*$ in $\Theta(|y|)$ time --- upper and lower bound follow from \cref{item:2} and \Cref{item:3} of the theorem statement, respectively. Using $\mathcal{A}$ and $\mathcal{B}$ we now construct Algorithm $\mathcal{D}$ that on an input $z \in \{0,1\}^N$ does the following:
\begin{enumerate}
    \item Run $\mathcal{A}$ on $z$ for $M=O(1)$ times and note the outputs.
    \item Run $\mathcal{B}$ on $z$ once.
    \item \texttt{Accept} if over $M$ runs of $\mathcal{A}$ on input $z$ the number of times $\mathcal{A}(z)=\mathcal{B}(z)$ is \emph{at most} $\frac{8}{10} \cdot M$, otherwise \texttt{reject}. We need the success probability of $\mathcal{D}$ to be at least $2/3$ therefore choosing $M=O(1)$ suffices. 
\end{enumerate}  
Notice that Algorithm $\mathcal{D}$ (with high probability) computes a property that is quasi-useful (\Cref{defn:Quasi-UsefulProperty}) against complexity class $\Lambda$. This is true because $\mathcal{B}(z)=\propertyP(z)$ for all $z$ and whenever $z \in \truthtable(\Lambda)$ the $\probabilityOf[\mathcal{A}(z)=\propertyP(z)] \geq 0.9$, which means the outputs $\mathcal{A}(z)$ and $\mathcal{B}(z)$ are equal more than $9/10$ of the $M$ times.  Therefore, if $z \in \truthtable(\Lambda)$ then $\mathcal{D}$ rejects $z$. 

What is left is to argue that Algorithm $\mathcal{D}$ computes a $\BQTIME{2\cdot N}$-natural property that satisfies both the \emph{largeness} and \emph{contructivity} conditions as mentioned in \Cref{defn:UpsilonNaturalProperty}. Let $z \in F_n$ and let $D_n=\{z \mid |z|=2^n \text{ and Algorithm } \mathcal{D} \text{ accepts } z \}$. 
\begin{itemize}
    \item Algorithm $\mathcal{D}$ on any $z$ of size $N=2^n$ can in at most $2 \cdot N$ quantum time decide whether or not $z \in D_n$; hence satisfying the \emph{constructivity} condition. 
    \item We argue about the largeness condition of $|D_n|\geq 2^N \cdot \delta_n$ for $\delta_n=\frac{1}{N^2}$ in the following way. Consider the set $F_n$, i.e., the set of truth tables of all Boolean functions defined on $n$ input variables. For a string $z \in F_n$, the algorithm $\mathcal{D}$ rejects $z$ only if $\probabilityOf[\mathcal{A}(z)=\propertyP(z)]> 8/10$. In other words, $\mathcal{D}$ rejects $z$ whenever $z \in F_n \setminus D_n$; see \Cref{fig:distinguisher}. The runtime of Algorithm $\mathcal{A}$ on $N$-length inputs is $O(N^{1-\beta})$ time (for a constant $\beta>0$), and we know from \cref{item:3} in the statement of the theorem that $\mathcal{A}$ cannot compute $\propertyP$ correctly (with success probability greater than $2/3$) on sets bigger than $2^N(1-\frac{1}{N^2})$. Hence, $|F_n \setminus D_n| \leq 2^N \cdot (1 - \frac{1}{N^2})$, which means $|D_n| \geq \frac{2^N}{N^2}$. 
\end{itemize}
This concludes the proof of $\cref{thm:LinearTimePropertyNotinFixedpCOimpliesQuasiUsefulProperty}$.
\end{proof}

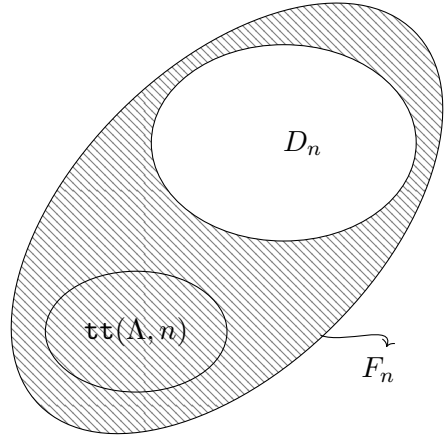
\begin{figure}
    \centering
\begin{tikzpicture}
\draw[rotate=45, pattern=north west lines, pattern color=gray] (0,0) ellipse (3.5cm and 2cm);
\draw[fill=white] (0.75,1) ellipse (1.75cm and 1.3cm);
\draw[pattern color=gray] (-1.2,-1.5) ellipse (1.2cm and 0.8cm);
\node at (2,-2) (A) {$F_n$};
\node at (-1.2,-1.5) (B) {$\truthtable(\Lambda,n)$};
\node at (1,1) (C) {$D_n$};
\node at (1.1,-1.5) (D) {};
\draw[->]   (D) to[out=-20,in=70] (A);
\end{tikzpicture}
    \caption{A pictorial representation of the fraction of $2^n$-length strings that the algorithm $\mathcal{D}$ accepts. Recall that $F_n$ denotes the full set $\{0,1\}^{2^n}$; hence, $|F_n|=2^{2^n}$. The set $\truthtable(\Lambda,n)$ denotes the set of strings $\{\truthtable(L,n) \mid L \in \Lambda \}$. Algorithm $\mathcal{A}$ is guaranteed to accept all strings from $\truthtable(\Lambda,n)$ with success probability at least $9/10$ and no guarantee whatsoever on strings in $F_n \setminus \truthtable(\Lambda,n)$. The set $D_n$ denotes the set of $2^n$-length strings that Algorithm $\mathcal{D}$ accepts.}
    \label{fig:distinguisher}
\end{figure}

\paragraph{Realising link (B)}
We will now show the cryptographic implications when properties like $\propertyParity$ and $\propertyMajority$ are not compression-oblivious. In link (A) we showed that, if either of these properties is not compression-oblivious with respect to complexity class $\Lambda=\SizeArg{d}{k}$ for some constants $d, k \in \mathbb{N}$, then we can construct a $\BQTIME{2 \cdot 2^n}$-natural property useful against $\Lambda$. Any quantum algorithm computing this natural property can now be used to distinguish pseudorandom functions constructible in $\Lambda=\SizeArg{d}{k}$ from truly random functions. However, such linear-time (linear in the length of the truth table of the functions) quantum distinguishers are not believed to exist. Therefore, it is plausible that a quantum algorithm computing this natural property in the specified time does not exist.

Let us recall the standard definition of pseudorandom functions~\cite{GGM86}.

\begin{defn}[Keyspace]
\label{def:Keyspace}
Let $K:\mathbb{N} \rightarrow \mathbb{N}$ be a function. A key space $\mathcal{K}$ generated by the function $K$ is a set of sets defined as $\mathcal{K}=\{ \{0,1\}^{K(n)} \}_{n \in \mathbb{N}}$ and we use $\mathcal{K}(n)$ to denote the set $\{0,1\}^{K(n)}$.

\end{defn}

Note that, for our results, the function $K$ is always chosen to be a polynomial.

\begin{defn}[Pseudorandom functions (PRFs)]
A set $F = \set{ f_k: \zo^n\to \zo }_{n\in\mathbb{N}, k\in \mathcal{K}(n),}$ is called a family of pseudorandom functions with key space $\mathcal{K}$ generated by function $K$ if for all probabilistic polynomial time (p.p.t.) adversary $A$, there is a negligible function $\epsilon$ such that for all $n\in\mathbb{N}$, 
\begin{equation*}
    |\Pr_{k}[A^{f_k(\cdot)}(1^n)=1]-\Pr_{\TR}[A^{\TR(\cdot)}(1^n)=1]| < \eta(n),
\end{equation*}
where the probability is also over the random string used by $A$, $\TR$ denotes a truly random function from $\zo^n$ to $\zo$.

Sometimes we may strengthen or weaken the definition by letting the adversary's running time be $T(n)$-bounded, and letting the distinguishing probability to be less than some specific function $\eta(n)$. In this case, we call it a $(K, T, \eta)$-$\PRF$. Additionally, if the family $F$, which can be interpreted as a families of circuits, corresponds to languages in $\Lambda$, then we call it as $\Lambda$-constructible $(K, T, \eta)$-$\PRF$s.
\end{defn}

Pseudorandom functions can also be defined against quantum adversaries making (classical or) quantum queries; see \cite{Zhandry12} for example. Formally stated as follows. 

\begin{defn}[Quantum-secure PRFs]
\label{def:QuantumSecurePRF}
A set $F = \set{ f_k: \zo^n\to \zo }_{k\in \mathcal{K}(n), n\in\mathbb{N}}$ is called a family of quantum-secure pseudorandom functions with key space $\mathcal{K}$ generated by a function $K$ if for all polynomial time (polynomial in the length of its input) quantum adversary $A$ making quantum queries, there is a negligible function $\epsilon$ such that for all $n\in\mathbb{N}$, 
\begin{equation*}
    |\Pr_{k}[A^{f_k(\cdot)}(1^n)=1]-\Pr_{\TR}[A^{\TR(\cdot)}(1^n)=1]| < \eta(n),
\end{equation*}
where the probability is also over the random string used by $A$, $\TR$ denotes a truly random function from $\zo^n$ to $\zo$.   


Sometimes (just as we do for $\PRF$s) we may strengthen or weaken the definition by letting the adversary's running time be $T(n)$-bounded, and letting the distinguishing probability governed by some other function $\eta(n)$. In this case, we denote them by $(K, T, \eta)$-quantum-secure-$\PRF$s. Additionally, if the family $F$, which can be interpreted as a families of circuits, correspond to languages in $\Lambda$, then we call it as $\Lambda$-constructible $(K, T, \eta)$-quantum-secure-$\PRF$s.
\end{defn}


Using this notion of quantum-secure pseudorandom functions, we are able to prove the following.

\begin{lemma}
\label{thm:QuasiusefulAgainstPpolyImpliesNoQsecurePRFs}
Let $d, k \in \mathbb{N}$. If there exists a $\BQTIME{2\cdot 2^n}$-natural property of density $\delta_n \geq \frac{1}{2^{2n}}$ that is quasi-useful against the complexity class $\Lambda=\SizeArg{d}{k}$ then there doesn't exist any $\Lambda$-constructible $(d \cdot n^k, 2\cdot 2^n, \frac{1}{2^{2n}})$-quantum-secure-$\PRF$; the size of the truth table of function $f_n$ is $2^n$.
\end{lemma}
\begin{proof}
Let $N=2^n$ be the length of the truth table of functions on $n$ variables. Let $\propertyP$ be a $\BQTIME{2 \cdot N}$-natural property of density $\delta_n \geq \frac{1}{2^{2n}}$ that is quasi-useful against the complexity class $\SizeArg{d}{k}$. Let $\mathcal{C}$ be the algorithm that computes this property $\propertyP$. Because $\mathcal{C}$ computes $\propertyP$, for any (sufficiently large) value of $n$ the following statements hold.
\begin{itemize}
    \item Algorithm $\mathcal{C}$ rejects $z$ whenever $z$ is a string from $\truthtable(\SizeArg{d}{k},n)$; this is because $\propertyP$ is quasi-useful against $\SizeArg{d}{k}$.
    \item There are at least $\frac{1}{2^{2n}} \cdot 2^{2^n}$ number of functions in $F_n$ that $\mathcal{C}$ accepts; this is because the $\propertyP$ has density $\delta_n \geq \frac{1}{2^{2n}}$.
    \item Furthermore, $\mathcal{C}$ runs in $\BQTIME{2 \cdot 2^n}$.
\end{itemize}
Combining these observations we can see that 
    \begin{equation*}
         |\probabilityOf[\mathcal{C}(x)=1]-\probabilityOf[\mathcal{C}(y)=1]| \geq \frac{1}{2^{2n}}
    \end{equation*}
when $x$ is chosen at random from $\truthtable(\SizeArg{d}{k},n)$ and $y$ is chosen at random from $F_n$. This means $\mathcal{C}$, despite being an algorithm running in $\BQTIME{2 \cdot N}$ time, can distinguish $\SizeArg{d}{k}$-constructible pseudorandom functions from the set of all random functions. Hence, showing that $\SizeArg{d}{k}$-constructible $(d \cdot n^k,2\cdot 2^n, \frac{1}{2^{2n}})$-quantum-secure-$\PRF$s do not exist.
\end{proof}

We could also state our assumption about the existence of quantum-secure $\PRF$s as a conjecture.

\begin{conjecture} 
\label{conj:QSecurePRF}
There exists $d,k \in \mathbb{N}$ for which $\SizeArg{d}{k}$-constructible $(d\cdot n^k,2\cdot 2^n, \frac{1}{2^{2n}})$-quantum-secure-$\PRF$s exist.
\end{conjecture}

Using \Cref{thm:LinearTimePropertyNotinFixedpCOimpliesQuasiUsefulProperty} and \cref{thm:QuasiusefulAgainstPpolyImpliesNoQsecurePRFs} we can now claim the following statement.

\begin{theorem}
\label{thm:LinearPviolatesQPRFassumption}
Let $\propertyP:\{0,1\}^{*} \rightarrow \{0,1\}$ be the Boolean property satisfying \cref{item:2,item:3} stated in \cref{thm:LinearTimePropertyNotinFixedpCOimpliesQuasiUsefulProperty}. Then property $\propertyP \in \mathcal{CO}(\Ppoly)$ unless \cref{conj:QSecurePRF} is false. 
\end{theorem}
\begin{proof}
If $\propertyP \notin \mathcal{CO}(\Ppoly)$ then $\forall d, \forall k \in \mathbb{N}$ we have $\propertyP \notin \mathcal{CO}(\SizeArg{d}{k})$; see \Cref{def:COforGeneralComplexityClass}. For every value of $d,k$ we will then be able to construct a $\BQTIME{2\cdot 2^n}$-natural property that will be quasi-useful against $\SizeArg{d}{k}$. This means for all values of $d,k$ we will no longer have $\SizeArg{d}{k}$-constructible $(d\cdot n^k, 2\cdot2^n, \frac{1}{2^{2n}})$-quantum-secure-$\PRF$s. Hence, falsifying \Cref{conj:QSecurePRF}.
\end{proof}

\begin{remark}



Although not explicitly used in this paper, we would also like to mention that under the assumption that the learning with errors (LWE) conjecture holds against quantum algorithms, there exist quantum-secure $\PRF$s evaluatable by $\NC^2$ circuits ($\NC^1$ if we assume the quantum hardness of RingLWE)~\cite{DBLP:conf/eurocrypt/BanerjeePR12}. This suggests that under the quantum hardness of LWE (resp. RingLWE), $\propertyParity$ and $\propertyMajority$ are  compression obliviousness with respect to the complexity class $\NC^2$ (resp. $\NC^1$).
\end{remark}



\begin{remark}[Strengthening \cref{conj:QSecurePRF} with OWFs] First, recall the definition of one-way functions.
\begin{defn}[One-way functions (OWFs)]
A function $g: \zo^*\to\zo^*$ is called a $(T, \eta)$-one-way function if there is a polynomial $p$ such that for all adversary $A$ running in time $T$, for all $n\in\mathbb{N}$, 
\begin{equation*}
    \Pr_{x\la \zo^{p(n)}}[A(1^n, g(x))=x' \text{ such that } g(x') = g(x)] < \eta(n).
\end{equation*}
\end{defn}    

For example, when $T = 2^n$, we can set the length of the domain and range to be $p(n) = n^2$, therefore a $T = 2^n$ adversary will not trivially break such an OWF.

It is known~\cite{HILL99,GGM86} that $(K, T, \eta)$-PRF exists if $(\poly(T), \poly(\eta))$-one-way function exists. Furthermore, the same result holds against quantum adversaries even if the adversary for the PRF is allowed to make superposition queries~\cite{Zhandry12}. Therefore, we can say that $(2^n, 2^{-\omega(n)})-$PRF against quantum adversaries exists as long as $(\poly(2^n), \poly(2^{-\omega(n)}))$-one-way function against quantum adversaries exists. Therefore, the following conjecture implies \Cref{conj:QSecurePRF}.
\begin{conjecture}
\label{conj:QSecureOWFs}
Quantum secure $(\poly(2^n),\poly(2^{-\omega(n)}))$-one-way functions exist.
\end{conjecture}
\end{remark}

\paragraph{Realising implication (C)}
Earlier in this section we discussed that the arguments presented work for properties that satisfy \cref{item:1,item:2,item:3} in the statement of \cref{thm:LinearTimePropertyNotinFixedpCOimpliesQuasiUsefulProperty}. For properties like $\propertyParity$ or $\propertyMajority$ it is trivial to see that they admit a linear time (quantum and classical) deterministic algorithm. However, it is not immediately clear if the query complexity of computing these properties is `high' on sets $S \subset \setOfAllStrings_n$ that are not the full set $\setOfAllStrings_n$ but still are relatively `large'.  Using techniques from Fourier analysis of Boolean properties, we can show that this is indeed the case with $\propertyParity$ and $\propertyMajority$. In fact, we are able to say something more. More precisely, we show that the following statement holds.

\begin{lemma}
\label{thm:WhichPropertiesAreHardEvenForSmallerSets}
Let $N$ be an integer and let $S \in \{0,1\}^N$ be a subset of size $s \geq 2^N\left(1-\frac{1}{N^{2}}\right)$. Let $\propertyP:\{0,1\}^N \rightarrow \{-1,1\}$ be a Boolean property with $Q(P)=\Omega(N)$ and the $N$-degree Fourier coefficient $|\fourierCoefficients{P}{[N]}| \geq \sqrt{\frac{1}{2N}}$.\footnote{Note that we use these constants so that \cref{thm:WhichPropertiesAreHardEvenForSmallerSets} can be almost directly used for $\propertyParity$ and $\propertyMajority$. However, one can modify these constants to accommodate other properties as well.} Then every quantum algorithm that, with error probability $\epsilon \leq \frac{1}{2\cdot N^{2}}$, computes $\propertyP$ on each $x \in S$ uses $\Omega(N)$ queries.    
\end{lemma}
\begin{proof} Suppose there exists a $T$-query algorithm, let us denote by $\mathcal{A}_S$, to compute property $\propertyP$ on all $x \in S$ with worst-case error probability $\epsilon \leq \frac{1}{2 \cdot N^{2}}$. W.l.o.g, we can assume that the output of $\mathcal{A}_S$ is in $\{-1,1\}$. Let $R(x)$ denote the probability that $\mathcal{A}_S$ outputs $-1$; then we can immediately see that $\forall x \in S$,
\begin{equation}
\label{eq:AlgorithmAgreesWithPAtManyPlaces}
    (1-2R(x))\cdot \propertyP(x) \geq (1-2\epsilon).\footnote{Since $R(x)$ is the probability that $\mathcal{A}_S$ outputs $-1$ and since the error probability of the algorithm is at most $\epsilon$, we know that (a) if $x\in P^{-1}(-1)$, then $R(x)\geq 1-\epsilon$, and (b) if $x\in P^{-1}(1)$, then $R(x)\leq \epsilon$. If this is the case (a), we obtain $(1-2R(x))\cdot P(x)=(1-2R(x))\cdot (-1)\geq 1-2\epsilon$. Moreover, even in case (b), we obtain $(1-2R(x))\cdot P(x)=(1-2R(x))\cdot (+1)\geq (1-2\epsilon)$. Hence \Cref{eq:AlgorithmAgreesWithPAtManyPlaces} holds.}
\end{equation}
We will now estimate the degree of the polynomial $R(x)$. Observe that $\text{deg}(R(x))=\text{deg}(1-2R(x))$. Let $Q(x)=(1-2R(x))$ and let $\epsilon'=(1-2\epsilon)$. Let the Fourier expansions of $Q(x)=\sum_{B \subseteq [N]}$ $\fourierCoefficients{Q}{B} \chi_{B}(x)$ and $\propertyP(x)=\sum_{B \subseteq [N]}$ $\fourierCoefficients{P}{B} \chi_{B}(x)$, respectively. 

Towards a contradiction, suppose that $\text{deg}(Q(x))<N$ then $\fourierCoefficients{Q}{[N]}=0$, which implies 
\begin{align*}
    \left(\sum_{B \in 2^{[N]}} \fourierCoefficients{Q}{B}\cdot \fourierCoefficients{P}{B}\right)^2 & \overset{(1)}{=} \left(\sum_{B \in 2^{[N]} \setminus [N]} \fourierCoefficients{Q}{B}\cdot \fourierCoefficients{P}{B}\right)^2\\ & \overset{(2)}{\leq} \left(\sum_{B \in 2^{[N]} \setminus [N]} \fourierCoefficients{Q}{B}^2\right)\cdot \left(\sum_{B \in 2^{[N]}\setminus [N]} \fourierCoefficients{P}{B}^2\right)\\ & \overset{(3)}{\leq} 1-\frac{1}{2N};
\end{align*}
$(1)$ using the assumption that $\fourierCoefficients{Q}{[N]}=0$, $(2)$ is argued using Cauchy-Schwarz inequality and $(3)$ comes by using Parseval's theorem (on Boolean output and bounded functions). However, from \Cref{eq:AlgorithmAgreesWithPAtManyPlaces} we have that $\forall x \in S, Q(x)\cdot \propertyP(x) \geq \epsilon'$ which means
\allowdisplaybreaks
\begin{align*}
\left(\sum_{B \in 2^{[N]}} \fourierCoefficients{Q}{B}\cdot \fourierCoefficients{P}{B}\right)^2 &= \left(\frac{1}{2^N}\sum_{x \in \{0,1\}^N} Q(x)\cdot \propertyP(x) \right)^2 \\ &= \frac{1}{2^N} \left( \sum_{x \in S}  Q(x)\cdot \propertyP(x) + \sum_{x \in \{0,1\}^N \setminus S} Q(x)\cdot \propertyP(x) \right)^2 \\
& \geq \left( \epsilon' \cdot \frac{|S|}{2^N} + \frac{1}{2^N}\sum_{x \in \{0,1\}^N \setminus S} Q(x)\cdot \propertyP(x)\right)^2 \\
& \geq \left(\epsilon' \cdot \frac{|S|}{2^N} - \frac{2^N -|S|}{2^N}\right)^2 \\
& \geq \left(\left(1-\frac{1}{N^{2}}\right) \cdot \frac{|S|}{2^N} - \frac{2^N -|S|}{2^N}\right)^2 \\ & \geq \left(1-\frac{3}{N^{2}}\right)^2 \\
& > 1-\frac{6}{N^{2}},
\end{align*}
leading to a contradiction; therefore, $\text{deg}(Q(x)) = \text{deg}(R(x))=N$; using polynomial method \cite{BBCMW01} this
implies $T=\Omega(N)$.
\end{proof}

Having stated \cref{thm:WhichPropertiesAreHardEvenForSmallerSets} we can now immediately prove similar query lower bounds for $\propertyParity$ and a \emph{variant} of $\propertyMajority$ as their $N$-degree Fourier coefficients are $1$ and at least $\sqrt{\frac{2}{\pi N}}$, respectively \cite{Ryan21Fourier}. Moreover, this kind of result can be extended to show similar lower bounds for all Boolean properties $\propertyP$ that have a non-negligible Fourier-mass of high degree coefficients. In fact, for $\propertyParity$ we can show that similar lower bounds hold even on smaller sets as well. Also, we can prove such query lower bounds for $\propertyMajority$ and $\propertyStrictMajority$. More precisely, 

\begin{claim}[see \cref{Appthm:ParityIsHardEvenForSmallerSets} in \cref{AppSec:ProofsRelevantToSupportingCOassumption} for proof]
\label{thm:parityIsHardEvenForSmallerSets}
 Let $S\in \{0,1\}^N$ be a subset of size $s\geq 0.8\cdot2^N$. Every quantum algorithm that, with success probability at least $\geq 2/3$, computes $\propertyParity$ on each $x \in S$ uses $\Omega(N)$ queries.
 \end{claim}
 
\noindent and,

\begin{claim}[see \cref{Appthm:MajorityAndStrictIsHardEvenForSmallerSets} in \cref{AppSec:ProofsRelevantToSupportingCOassumption} for proof]
\label{thm:majorityAndStrictIsHardEvenForSmallerSets}
Let $S \in \{0,1\}^N$ be a subset of size $s \geq 2^N\left(1-\frac{1}{2 \cdot N^2}\right)$. Every quantum algorithm that, with success probability $(1-\epsilon) \geq (1-\frac{1}{2\cdot N^2})$, computes $\propertyMajority$ ($\propertyStrictMajority$) on each $x \in S$ uses $\Omega(N)$ queries.
\end{claim}

The result proved in \cref{thm:LinearTimePropertyNotinFixedpCOimpliesQuasiUsefulProperty} when combined with the results stated in \Cref{thm:parityIsHardEvenForSmallerSets,thm:majorityAndStrictIsHardEvenForSmallerSets} along with \cref{thm:LinearPviolatesQPRFassumption} lead us to the following result.

\begin{corollary}
\label{thm:ParityMajorityviolateQPRFassumptionIfNotCO}
Let $\Lambda=\Ppoly$ be the complexity class. The properties $\propertyParity \in \mathcal{CO}(\Lambda)$, $\propertyMajority \in \mathcal{CO}(\Lambda)$ and $\propertyStrictMajority \in \mathcal{CO}(\Lambda)$ unless \cref{conj:QSecurePRF} is false.
\end{corollary}

\bibliographystyle{alpha}
\bibliography{Lattice.bib}

\appendix

\section{Observations on compression-oblivious properties made in \cite{BPS21}}
\label{Appsec:OldObservationsAboutCO}

To ensure that our detailed definition of compression oblivious is valid, and is what the authors of \cite{BPS21} had in mind, we reprove their results about compression oblivious properties again here. The proof ideas are almost the same, the extra details in our proofs are to make their proofs consistent with our proposed definitions of compression oblivious properties.

\begin{lemma}[Example~6 in \cite{BPS21}]
\label{thm:ORisCOforAC}
Let $\propertyOR:\{0,1\}^* \rightarrow \{0,1\}$ be a Boolean property such that $\forall x \in \{0,1\}^*$ the $\propertyOR(x)=1$ if and only if $|x|\geq 1$. Then $\propertyOR \in \mathcal{CO}(\AC)$. Here $|x|$ denotes the hamming weight of $x$, i.e., the number of $1$s in $x$.
\end{lemma}
The proof idea is to construct a hard set of languages for every $n \in \mathbb{N}$ that makes it difficult to compute $\propertyOR_n$ in a black-box way when given access to the set of $n$-input circuits each corresponding to a language in this hard set of languages. 
\begin{proof}
To prove that $\propertyOR \in \mathcal{CO}(\AC)$, we show that there exists a $k' \in \mathbb{N}$ such that $\propertyOR \in \mathcal{CO}(\ACarg{2}{k'})$. For every value of $n \in \mathbb{N}$ we will exhibit a set $L = \{L^0, L^1, \ldots \} \subseteq \ACarg{2}{k'}$ whose respective truth tables, i.e., the set $\{ \truthtable(L^i,n)\}_{i \in [|L|]}$, form a hard set for computing $\propertyOR_n$ on them. In fact, the value of $k'$ in our construction will be $2$.

\paragraph{The construction for $L$.} Let $n \in \mathbb{N}$. Let $L^0=\emptyset$. We will construct the languages $L^1, L^2, \ldots$ so that each language $L^i$ contains at most one Boolean string of length $k$.
\begin{enumerate}
    \item For every $k \in [n]$, starting with $k=1$ in an increasing order, go over all the strings $z \in \{0,1\}^k$ in a lexicographical order. Place the string $z$ in set $L^i$ (for an $i>0$) only if no $L^j$ with $j \leq i$ contains the string $z$. Once the placement of $z$ is decided, set $z$ to be the next string in the lexicographic ordering of $\{0,1\}^k$. 
    \item For $k> n$, we put no $k$-length binary strings in any of these $L^i$s.
\end{enumerate}

\paragraph{All these $L^i$s are in $\ACarg{2}{2}$.} For every $i$, let $\{C^i_n\}_{n \in \mathbb{N}}$ denote a family of circuits corresponding to $L^i \in L$. We will now show that there exist families such that $|C^i_n| <n^2$ for all $n \in \mathbb{N}$. 
\begin{enumerate}
    \item It is easy to see that the $\{C^0_n\}_{n \in \mathbb{N}}$ can be constructed using circuits of constant size as $L^0$ has no accepting strings of any length. Therefore, for all $n \in \mathbb{N}$, $|C^0_n|<n^2$.
    \item Now for the other $L^i$s:
    \begin{enumerate}
        \item For an $k <=n$, for any $L^i$ with $i>0$, the corresponding $C^i_k$ accepts exactly one $k$-length string and such a circuit can be constructed using one $\land$ gate with $k$-fanin. Therefore, $|C^i_k|=k+1<k^2$. (For computing the size, we count the number of vertices in the circuit.)
        \item For an $k >n$, $C^i_k$ can be constructed using constant size circuits as there are no accepting $k$-length strings for any $L^i$.
    \end{enumerate}
    Therefore, for all $n\in \mathbb{N}$ and for all $i \in [|L|]$ we have that $|C^i_n| < n^2$, which means $L \subseteq \ACarg{2}{2}$.
\end{enumerate}

\paragraph{Invoking the quantum query lower bound for $\propertyOR_n$.} What remains to show is that no quantum algorithm can compute $\propertyOR_n$ on the strings $\{ \truthtable(L^i, n)\}_{i \in [|L|]}$ in $\Q_{1/3}(\propertyOR_n)^{1-\delta}$ time for any constant $\delta>0$. Towards contradiction, say that such a $\delta$ existed. Then this algorithm can compute $\propertyOR_n$ on $2^n$-length strings that have hamming weight at most $1$ in $2^{\frac{n(1-\delta)}{2}}$ time, hence in $2^{\frac{n(1-\delta)}{2}}$ queries. This is not possible because using the quantum adversary method \cite{Ambainis02-QuantumAdversaryMethod} one can show that any quantum algorithm computing $\propertyOR_n$ on these strings requires at least $c \cdot 2^{\frac{n}{2}}$ many queries.

Combining all these observations, we can conclude that $\propertyOR \in \mathcal{CO(\ACarg{2}{2})}$. Therefore, using \Cref{def:COforGeneralComplexityClass},
we get $\propertyOR \in \mathcal{CO(\AC)}$.
\end{proof}

One can use similar arguments to show that $\propertyAND$ is compression-oblivious with respect to $\ACconstDepth$.

\begin{corollary}[Example~6 in \cite{BPS21}]
Let $\propertyAND:\{0,1\}^* \rightarrow \{0,1\}$ be a Boolean property such that $\forall x \in \{0,1\}^*$ the $\propertyAND(x)=1$ if and only if $|x|=\mathsf{len}(x)$. Then $\propertyAND \in \mathcal{CO}(\AC)$. Here $\mathsf{len}(x)$ denotes the length of the binary string $x$.  
\end{corollary}

\begin{lemma}[Fact~15 in \cite{BPS21}]
\label{thm:CompressionObliviousProperties} 
Let $\zeta$ and $\gamma$ be two complexity classes in $\{\AC, \ACconstDepth, \NC, \Ppoly \}$. Furthermore, let $\zeta \subseteq \gamma$, then for every property $\propertyP$, we have $\propertyP \in \mathcal{CO}(\gamma)$ whenever $\propertyP \in \mathcal{CO}(\zeta)$.   
\end{lemma}
\begin{proof}
Without loss of generality, let us assume that $\zeta$ is $\ACconstDepth$.
If a property $\propertyP \in \mathcal{CO}(\ACconstDepth)$ then there exists $d, k \in \mathbb{N}$ such that $\propertyP \in \mathcal{CO}(\ACarg{d}{k})$; see \Cref{def:COforGeneralComplexityClass}. From \Cref{def:SizeableComplexityClass} we know that for every value of $d, k$ we can constructively (by constructing explicit circuits) show that there exists $d',k'$ and $d'',k''$ such that $\ACarg{d}{k} \subseteq \NCarg{d'}{k'} \subseteq \SizeArg{d''}{k''}$. Therefore, any set of hard languages in $\ACconstDepth$ that witnesses that $\propertyP$ is compression-oblivious with respect to $\ACconstDepth$ will also witness that $\propertyP$ is compression-oblivious with respect to $\{\NC, \Ppoly \}$. Therefore, $\propertyP \in \mathcal{CO}(\gamma)$ for any $\gamma \in \{ \NC, \Ppoly\}$. We can make a similar argument for other combinations of $\zeta, \gamma \in \{\ACconstDepth, \NC, \Ppoly \}$ satisfying $\zeta \subseteq \gamma$.
\end{proof}

We saw that properties such as $\propertyOR$ and $\propertyAND$ are compression-oblivious with respect to $\ACconstDepth, \NC, \Ppoly$. One can also construct properties that are not compression-oblivious for any of these complexity classes.

\begin{lemma}[Example~7 in \cite{BPS21}] Consider the following Boolean property $\propertyP_{\mathtt{large-c}}(z)=(\propertyP_{\mathtt{large-c},n}(z))_{n \in \mathbb{N}}$ that for all $z \in \{0,1\}^{2^n}$,
\begin{equation*}
    \propertyP_{\mathtt{large-c},n}(z) = \propertyParity_{n}(z) \land [\text{$\nexists$ circuit $C$ on $n$ inputs of size less than $2^{\frac{n}{100}}$ s.t. $z = \truthtable(C)$}].
\end{equation*}
The property $\propertyP_{\mathtt{large-c}}$ is not compression-oblivious with respect to $\{\ACconstDepth, \NC, \Ppoly\}$.
\end{lemma}
\begin{proof}
First, notice that the query complexity of $\propertyP_{\mathtt{large-c},n}$ is very high; this is because most strings are not a truth table of small circuits, the query complexity of this property is close to the query complexity of $\propertyParity_n$, i.e., $\Q_{1/3}(\propertyP_{\mathtt{large-c}})=\Omega(2^n)$. Moreover, there exists an algorithm $\mathcal{A}$ that can compute $\propertyP_{\mathtt{large-c}}$ in $0$ queries no matter the length of the input: $\mathcal{A}$ always outputs $0$. This fixes $\delta=1$. When $\Lambda=\{\ACconstDepth, \NC, \Ppoly\}$, $\forall n \in \mathbb{N}$, Algorithm $\mathcal{A}$ computes $\propertyP_{\mathtt{large-c}}$ on all strings in $\truthtable( \Lambda, n)$ in constant time. Hence, $\propertyP_{\mathtt{large-c}}$ is not compression-oblivious for $\Ppoly$ or any smaller class of representations.
\end{proof}

Now, moving on to more complexity-theoretic consequences.

\begin{theorem}[Theorem~9 in \cite{BPS21}] If there exists a property $\propertyP=(\propertyP_n)_{n \in \mathbb{N}}$ and constant $p>0$ s.t.\ $\Q_{1/3}(\propertyP_{n}) = N^{\frac{1}{2}+p}$ and $\propertyP \in \mathsf{polyL}(N)$ and $\propertyP \in \mathcal{CO}(\Ppoly)$, then $\P \neq \PSPACE$. Here $N=2^n$.\footnote{\label{footnote:SlightChangePropertyConditions} Buhrman \textit{et al.}, under their notion of compression oblivious properties proved this theorem for properties $\propertyP=(\propertyP_n)_{n \in \mathbb{N}}$ that are in $\mathsf{polyL}(N)$ and satisfy $\Q_{1/3}(\propertyP_{n}) = \widetilde{\omega}(\sqrt{N})$ but with our proposed definition of compression oblivious we need that the properties satisfy $\Q_{1/3}(\propertyP_{n}) = N^{\frac{1}{2}+p}$ for some constant $p>0$ that only depends on $\propertyP$.}
\end{theorem}
Here $\mathsf{polyL}(N)$ is the same as $\mathsf{SPACE}(\polylog (N))$, i.e., class of properties $\propertyP$ (or the corresponding languange $L_{\propertyP}$) computable in $\polylog(N)$ amount of space. 
\begin{proof}
We prove the contrapositive of this statement. Let $d, k \in \mathbb{N}$ and let $\Lambda=\SizeArg{d}{k}$. Let $\propertyP=(\propertyP_n)_{n \in \mathbb{N}}$ be a property satisfying $\Q_{1/3}(\propertyP_{n}) = N^{\frac{1}{2}+p}$ and $\propertyP \in \mathsf{polyL}(N)$ and suppose that $\P =\PSPACE$ then we can show that  $\propertyP \notin \mathcal{CO}(\SizeArg{d}{k})$ for all $d, k \in \mathbb{N}$. Let us first fix a pair $d,k \in \mathbb{N}$, under the assumption of $\P = \PSPACE$ for this $\propertyP$ we show that there exists a $c,n' \in \mathbb{N}$ and we can construct an algorithm $\mathcal{A}$ such that $\forall n>n'$ for all inputs in $\{0,1\}^{N}$ the algorithm $\mathcal{A}$ computes $\propertyP_n$ in a black-box manner in at most $c \cdot d \cdot n^k \cdot \sqrt{N}$ time. (Algorithm $\mathcal{A}$ is basically the oracle identification algorithm of \cite{Kothari14} which can be made time efficient by using $\P=\PSPACE$.) As all the properties $\propertyP$ in consideration have their respective $\Q_{1/3}(\propertyP_{n}) = N^{\frac{1}{2}+O(1)}$ therefore for each property $\propertyP$ there is corresponding $\delta>0$ such that $c \cdot d \cdot n^k \cdot \sqrt{N} = \Q_{1/3}(\propertyP_n)^{1-\delta}$ for all $n > n'$. 

We can do this for every property $\propertyP = (\propertyP_n)_{n \in \mathbb{N}}$ satisfying $\Q_{1/3}(\propertyP_{n}) = N^{\frac{1}{2}+O(1)}$ and $\propertyP \in \mathsf{polyL}(N)$. The values of $c, n'$ and $\delta$ differ with each property but the algorithm $\mathcal{A}$ is the same. For the construction of Algorithm $\mathcal{A}$, we refer the reader to the proof of Theorem~9 in \cite{BPS21}. 

As we can do this for every pair $d, k \in \mathbb{N}$ we can (under the assumption of $\P=\PSPACE$) show that $\propertyP \notin \mathcal{CO}(\Ppoly)$. Hence, proved.
\end{proof}

\begin{theorem}[Theorem~10 in \cite{BPS21}] There exists an oracle relative to which the basic-\QSETH{} holds, but any property $\propertyP=(\propertyP_n)_{n \in \mathbb{N}}$ with $\propertyP \in \mathsf{polyL}(N)$ for which $\Q_{1/3}(\propertyP_{n}) = N^{\frac{1}{2}+O(1)}$ is not compression oblivious with respect to $\Ppoly$, i.e., $\propertyP \notin \mathcal{CO}(\Ppoly)$. Here $N=2^n$.\footnote{See \Cref{footnote:SlightChangePropertyConditions}.}
\end{theorem}
\begin{proof} The idea is to first prove this statement for the complexity class $\SizeArg{d}{k}$ for a fixed pair $d, k \in \mathbb{N}$ and then show that we can repeat the argument for every $d,k \in \mathbb{N}$. Having once chosen a pair $d,k$, we construct the oracle \emph{exactly} as it is constructed in the proof of Theorem~10 in \cite{BPS21}. The only part where our proof differs is in the assumption that $\Q_{1/3}(\propertyP_{n}) = N^{\frac{1}{2}+O(1)}$ instead of $\Q_{1/3}(\propertyP_{n}) = \widetilde{\omega}(\sqrt{N})$ as originally done in the statement of Theorem~10 in \cite{BPS21}; this change makes this theorem true with the definition of compression oblivious properties we propose in this paper. 
\end{proof}

\section{Proofs of \cref{thm:parityIsHardEvenForSmallerSets,thm:majorityAndStrictIsHardEvenForSmallerSets}}
\label{AppSec:ProofsRelevantToSupportingCOassumption}

\begin{lemma}
\label{Appthm:ParityIsHardEvenForSmallerSets}
Let $S\in \{0,1\}^N$ be a subset of size $s\geq 0.8\cdot2^N$. Every quantum algorithm that, with success probability at least $\geq 2/3$, computes $\propertyParity$ on each $x \in S$ uses $\Omega(N)$ queries.
\end{lemma}
\begin{proof}
 Suppose there exists a $T$-query bounded-error quantum algorithm $\mathcal{A}_S$ to compute parity on each $x \in S$ with worst-case error probability $\leq 1/3$. Here, without loss of generality, we assume the output of $\mathcal{A}_S$ is $\in \{-1,1\}$, where $1$ means odd parity and $-1$ means even parity. Let $P(x)$ be the polynomial that represents $\mathcal{A}_S$'s probability of outputting $-1$ on input $x$. Since $\mathcal{A}_S$ has worst-case error probability at most $1/3$ on each $x\in S$, we can immediately see $(1-2P(x))\cdot (-1)^{|x|}\geq 1/3$ for every $x\in S$, where $|x|$ denotes the Hamming weight of $x$.

 Observe that $\text{deg}(P(x))=\text{deg}(1-2P(x))$ and  $\E_{x\in \{0,1\}^N}[(1-2P(x))(-1)^{|x|}]$ is the Fourier coefficient of the degree-$N$ term of $(1-2P(x))$. We obtain
 \begin{align*}
     \E_{x\in \{0,1\}^N}[(1-2P(x))(-1)^{|x|}] =& \frac{1}{2^N}\big(\sum_{x\in S}[(1-2P(x))(-1)^{|x|}]+\sum_{x\in \{0,1\}^N\setminus S}[(1-2P(x))(-1)^{|x|}]\big)\\
     \geq & \frac{1}{2^N}\cdot \big(s/3-(2^N-s) \big)>0,
 \end{align*}
 implying that $\text{deg}(P(x))\geq N$. 
 Then by polynomial method~\cite{BBCMW01},~\footnote{By the polynomial method we know if the amplitudes of the final state are degree-$N$ polynomials of $x$ then we need to make at least $T\geq N/2$ queries.} we immediately obtain $T\geq N/2$.
\end{proof}

\begin{lemma}[Theorem~5.19 in \cite{Ryan21Fourier}]
\label{thm:FourierCoefficientsOfMajority}
Let $N$ be an odd integer and let $\propertyOddMajority:\{0,1\}^N \rightarrow \{0,1\}$ be defined as 
\begin{equation*}
\propertyOddMajority(x)=
    \begin{cases}
    1, & \text{if } |x|>\frac{N}{2},\\
    0, & \text{otherwise},
\end{cases}    
\end{equation*}
where $|x|=|\{i \mid i \in [N] \text{ and } x_i=1\}|$. If the Fourier expansion of $\propertyOddMajority$ is 
\begin{equation*}
\propertyOddMajority(x)=\sum_{B \subseteq \{0,1\}^N} \fourierCoefficients{odd-maj}{B} \chi_{B}(x),
\end{equation*}
where $\chi_{B}(x)=(-1)^{x.B}$, then $\fourierCoefficients{odd-maj}{B}=0$ whenever $|B|==0 \mod 2$ and for $|B|=N$ we have
\begin{equation*}
    \fourierCoefficients{odd-maj}{B}^2=\frac{4}{4^N} {N-1 \choose \frac{N-1}{2}}^2 \geq \frac{6}{\pi (3N-1)} > \frac{2}{\pi N}.
\end{equation*}
\end{lemma}

\begin{lemma}
\label{Appthm:OddMajorityIsHardEvenForSmallerSets}
Let $N$ be an odd integer and let $S \in \{0,1\}^N$ be a subset of size $s \geq 2^N\left(1-\frac{1}{N^{2}}\right)$. Every quantum algorithm that, with error probability $\epsilon \leq \frac{1}{2\cdot N^{2}}$, computes $\propertyOddMajority$ on each $x \in S$ uses $\Omega(N)$ queries.    
\end{lemma}
\begin{proof}
The absolute value of the $N$-degree Fourier coefficient for $\propertyOddMajority$ is at least $\sqrt{\frac{2}{\pi N}}$ using \cref{thm:WhichPropertiesAreHardEvenForSmallerSets} gives us the required proof.
\end{proof}

\begin{lemma}
\label{Appthm:EvenMajorityIsAlsoHardOnSmallSets}
Let $N$ be an even integer and let $S \in \{0,1\}^N$ be a subset of size $s \geq 2^N(1-\frac{1}{2\cdot N^2})$. Every quantum algorithm that, with error probability $\epsilon \leq \frac{1}{2\cdot N^{2}}$, computes $\propertyMajority$ on each $x \in S$ uses $\Omega(N)$ queries.  
\end{lemma}
We embed the $\propertyOddMajority$ into the $\propertyMajority$ and then argue lower bound for $\propertyMajority$ for sets of size at least $2^N(1-\frac{1}{2\cdot N^2})$  using the result obtained in \cref{Appthm:OddMajorityIsHardEvenForSmallerSets}.
\begin{proof}
Towards contradiction, let us assume that there is a set $B \subseteq \{0,1\}^N$ of size $b \geq 2^N(1-\frac{1}{2\cdot N^2})$ that uses at most $N^{0.999}$ queries to compute $\propertyMajority$ on all strings $y \in B$. The set $B=\{B_0 \concatenate 0\} \sqcup \{B_1 \concatenate 1\}$ can be viewed as two disjoint sets of strings ending with $0$ and $1$, respectively. Clearly $|B_1| \leq 2^{N-1}$ and by our assumption we have $|B| \geq 2^N(1-\frac{1}{2 \dot N^2})$, therefore, $|B_0|=|B|-|B_1| \geq 2^N(\frac{1}{2}-\frac{1}{2 \cdot N^2})$. Set $S'=B_0$. The set $S'$ is a subset of strings of $N'$-length strings where $N'=N-1$ and moreover $N'$ is odd. Therefore, in terms of $N'$ we have $|B_0|\geq 2^{N'+1}(\frac{1}{2}-\frac{1}{2 \cdot (N'+1)^2})=2^{N'}-\frac{2^{N'}}{\frac{2}{2} \cdot (N'+1)^2} > 2^{N'}(1-\frac{1}{{N'}^2})$. As $S'=B_0$, computing $\propertyMajority$ on strings in $B_0 \concatenate 0$ (or equivalently $S' \concatenate 0$) would be equivalent to computing $\propertyOddMajority$ on $N'$-length strings in $S'$ of size $|S'| \geq 2^{N'}(1-\frac{1}{{N'}^2})$. But we know from \cref{Appthm:OddMajorityIsHardEvenForSmallerSets} a sublinear query algorithm is not possible for these set sizes. Hence, proved.
\end{proof}

We can give query lower bound for $\propertyStrictMajority$ by just setting $S'=B_1$ and running the same arguments mentioned in the proof of \cref{Appthm:EvenMajorityIsAlsoHardOnSmallSets}. Giving the following result for $\propertyStrictMajority$.

\begin{lemma}
\label{Appthm:EvenStrictMajorityIsAlsoHardOnSmallSets}
Let $N$ be an even integer and let $S \in \{0,1\}^N$ be a subset of size $s \geq 2^N(1-\frac{1}{2\cdot N^2})$. Every quantum algorithm that, with error probability $\epsilon \leq \frac{1}{2\cdot N^{2}}$, computes $\propertyStrictMajority$ on each $x \in S$ uses $\Omega(N)$ queries.  
\end{lemma}

When $N$ is odd, all the three definitions of $\propertyMajority$, $\propertyStrictMajority$ and $\propertyOddMajority$ are equivalent. Therefore, combining \cref{Appthm:OddMajorityIsHardEvenForSmallerSets,Appthm:EvenMajorityIsAlsoHardOnSmallSets,Appthm:EvenStrictMajorityIsAlsoHardOnSmallSets} we get the following.

\begin{lemma}
\label{Appthm:MajorityAndStrictIsHardEvenForSmallerSets}
Let $S \in \{0,1\}^N$ be a subset of size $s \geq 2^N\left(1-\frac{1}{2 \cdot N^2}\right)$. Every quantum algorithm that, with success probability $(1-\epsilon) \geq (1-\frac{1}{2\cdot N^2})$, computes $\propertyMajority$ ($\propertyStrictMajority$) on each $x \in S$ uses $\Omega(N)$ queries.
\end{lemma}

\end{document}